%% file: main.tex
\documentclass{article}

\usepackage[english]{babel}
\usepackage{amsthm,amsmath,lipsum}
\usepackage[ruled]{algorithm2e}
\usepackage{xcolor}
\usepackage{fullpage}
\usepackage{amssymb}
\usepackage{centernot}
\usepackage{authblk}
\usepackage{tikz}
\usepackage{soul}
\usepackage{natbib}
\usetikzlibrary{shapes.geometric, arrows}
\usetikzlibrary{positioning}
\usetikzlibrary{calc}

\DeclareFontFamily{U}{mathb}{\hyphenchar\font45}
\DeclareFontShape{U}{mathb}{m}{n}{
      <5> <6> <7> <8> <9> <10> gen * mathb
      <10.95> mathb10 <12> <14.4> <17.28> <20.74> <24.88> mathb12
}{}
\DeclareSymbolFont{mathb}{U}{mathb}{m}{n}
\DeclareMathSymbol{\llcurly}{3}{mathb}{"CE}
\DeclareMathSymbol{\ggcurly}{3}{mathb}{"CF}

\usepackage{thm-restate}

\usepackage[colorlinks=true,linkcolor=black,anchorcolor=black,citecolor=black,filecolor=black,menucolor=black,runcolor=black,urlcolor=black]{hyperref}
\usepackage[capitalise,nameinlink
    ,noabbrev]{cleveref}

\input{Macros.tex}

\theoremstyle{plain}
\newtheorem{theorem}{Theorem}[section]
\newtheorem{definition}[theorem]{Definition}
\newtheorem{lemma}[theorem]{Lemma}
\newtheorem{corollary}[theorem]{Corollary}
\newtheorem{observation}[theorem]{Observation}
\newtheorem{claim}[theorem]{Claim}
\newtheorem{remark}[theorem]{Remark}

\tikzstyle{startstop} = [rectangle, rounded corners, minimum width=2.5cm, minimum height=1cm,text centered, draw=black, fill=red!30]
\tikzstyle{io} = [trapezium, trapezium left angle=70, trapezium right angle=110, minimum width=3cm, minimum height=1cm, text centered, draw=black, fill=blue!30]
\tikzstyle{process} = [rectangle, minimum width=2.5cm, minimum height=1cm, text centered, draw=black, fill=blue!30]
\tikzstyle{decision} = [diamond, minimum width=3cm, minimum height=1cm, text centered, draw=black, fill=green!30]
\tikzstyle{arrow} = [thick,->,>=stealth]

\tikzstyle{diam} = [diamond, aspect=2, draw, fill=red!40, text width=6em,text centered ]
\tikzstyle{block} = [rectangle, draw, fill=blue!20, text width=3cm,text centered, rounded corners, minimum height=2em ]
\tikzstyle{trap} = [trapezium, trapezium left angle=70, trapezium right angle=110, minimum height=2em, text centered, draw=red, fill=green!30]
\tikzstyle{rect} = [rectangle, minimum width=3cm, minimum height=1cm, text centered, draw=red, fill=orange!30]
\tikzstyle{line} = [draw, -latex]

\title{EF$2$X Exists For Four Agents}
\author[1]{Arash Ashuri\thanks{arash.ashoori0330@sharif.edu}}
\author[2]{Vasilis Gkatzelis\thanks{gkatz@drexel.edu. This author was partially supported by  NSF CAREER award CCF-2047907.}}
\author[3]{Alkmini Sgouritsa\thanks{alkmini@aueb.gr. The research project is implemented in the framework of H.F.R.I call “Basic research Financing (Horizontal support of all Sciences)” under the National Recovery and Resilience Plan “Greece 2.0” funded by the European Union-–NextGenerationEU (H.F.R.I. Project Number:15635).}}
\affil[1]{Sharif University of Technology}
\affil[2]{Drexel University}
\affil[3]{Athens University of Economics and Business, and Archimedes/Athena RC}
\date{}

\begin{document}
\maketitle

\begin{abstract}
    We study the fair allocation of indivisible goods among a group of agents, aiming to limit the envy between any two agents. The central open problem in this literature, which has proven to be extremely challenging, is regarding the existence of an EFX allocation, i.e., an allocation such that any envy from some agent $i$ toward another agent $j$ would vanish if we were to remove any single good from the bundle allocated to $j$. When the agents' valuations are additive, which has been the main focus of prior works, \cite{CGM24} were able to show that an EFX allocation is guaranteed to exist for all instances involving up to three agents. Subsequently, \cite{BCFF21} extended this guarantee to nice-cancelable valuations and \cite{ACGMM22} to MMS-feasible valuations. However, the existence of EFX allocations for instances involving four agents remains open, even for additive valuations.

    We contribute to this literature by focusing on EF2X, a relaxation of EFX which requires that any envy toward some agent vanishes if any \emph{two} of the goods allocated to that agent were to be removed. Our main result shows that EF2X allocations are guaranteed to exist for any instance with four agents, even for the class of cancelable valuations, which is more general than additive. Our proof is constructive, proposing an algorithm that computes such an allocation in pseudopolynomial time. Furthermore, for instances involving three agents we provide an algorithm that computes an EF2X allocation in polynomial time, in contrast to EFX, for which the fastest known algorithm for three agents is only pseudopolynomial.
\end{abstract}

\section{Introduction}
    During the last decade, the main focus of the fair division literature has been on the fair allocation of indivisible goods, and one of the main goals within this line of work is to limit the amount of envy between any two agents. Since it is well-known that some amount of envy is, in some cases, inevitable when allocating indivisible goods,\footnote{E.g., if one of the goods is more valuable than all other goods combined, any agent who receives it is bound to be envied.} a vast amount of work has focused on relaxations of ``envy-freeness'' and the extent to which they can be achieved or approximated. 
    
    A central problem within this literature is regarding the existence of allocations satisfying the EFX property, introduced by \citet{CKMPSW19}. An allocation is EFX if removing any \itm\ from the bundle allocated to an agent would ensure that no agent envies the remaining bundle. Despite a plethora of attempts to prove the existence (or non-existence) of EFX allocations in general, this problem remains open. \citet{PR20} proved the existence of EFX allocations for two agents.
    For instances with three agents, \citet{CGM24} were able to prove the existence of EFX allocations for additive valuations, and this was subsequently extended to nice-cancelable valuations (a special case of cancelable valuations) by \citet{BCFF21} and 
    even more general valuations by \cite{ACGMM22}. 
    In a very recent independent work that was done in parallel with our work,
    \cite{PGNV24} proved EFX allocations exist for any number of agents when there are at
    most three types of additive valuations.
    The existence of EFX allocations for instances involving four agents remains a very challenging open problem, even in the case of additive valuations.

    Our work focuses on an interesting relaxation of EFX, namely EF2X.
    An allocation is EF2X if removing any \emph{two} goods from the bundle allocated to an agent 
    would ensure that no agent envies the remaining bundle.
    This notion was introduced by \citet{ARS22}, who proved that an EF2X allocation exists for instances involving any number of agents, but with valuations that are rather restrictive: an agent's value for a bundle equals the sum of their values for each good in the bundle, like additive valuations, but the agent's value for each good $g$ is restricted to be either $0$ or $v(g)$, where the latter value is the same for all agents. In very recent work, \citet{KSS24} proved the same result for a different class of rather restrictive valuations ($(\infty,1)$-bounded valuations): this class defines a \itm\ to be ``relevant'' for an agent if its marginal value is not always zero, and the crucial restriction is that any two agents share at most one relevant \itm.

\subsection{Our Results}
    Rather than significantly restricting the valuations that the agents may have, we focus on valuations that are much more expressive, and we study the existence of EF2X allocations for instances involving four agents. Specifically, we consider \emph{cancelable valuations}, introduced by~\citet{BCFF21}, which strictly generalize the well-studied class of additive valuations. A valuation $v$ over a set of \itms\ $M$ is cancelable if for any two bundles $S, T \subset M$ and any \itm \ $g \in M \setminus (S \cup T)$, we have that $v (S \cup \{g\}) > v (T \cup \{g\})$ implies $v (S) > v (T)$. 
        
    Our main result uses a constructive argument to prove that EF2X allocations always exist for any instance involving four agents with cancelable valuations.
        
    \vspace{0.1in}
    \noindent {\bf Theorem:} For every instance involving four agents with cancelable valuation functions and any number of goods there exists an EF2X allocation.
    \vspace{0.1in}
        
    In fact, we only require three of these agents to have cancelable valuations; our result holds even if one agent has arbitrary monotone valuations. 
    Furthermore, we show that we can compute such an allocation in pseudopolynomial time.
    
    Proving this result is quite demanding, as it inherits a lot of the complexity of the EFX problem. Note that proofs of EFX existence for three agents tend to be quite long, requiring non-trivial techniques and careful case analysis. As a result, extending those arguments to instances with four agents becomes rather intractable. We overcome this obstacle by introducing some new notions and techniques that are likely to be useful more broadly, e.g., for solving the EFX problem (in fact, the allocation that we compute is often EFX, not just EF2X). These contributions allow us to make the problem more tractable, but our algorithm and its analysis are still quite long and technically demanding. 
    
    As a secondary result, using the same tools, for instances involving only three agents, we provide an algorithm that terminates in polynomial time.
        
    \vspace{0.1in}
    \noindent {\bf Theorem:} For every instance involving three agents with cancelable valuations and any number of goods, we can compute an EF2X allocation in polynomial time.

\subsection{Additional Related Work}
    EFX allocations have been shown to exist for any number of agents when these agents have identical valuations \citep{PR20}, lexicographic preferences \citep{HSVX21}, additive valuations with at most two types of goods \citep{GMV23}, as well as binary valuations \citep{HPPS20}, subsequently extended to bi-valued valuations \citep{ABRHV21}. \citet{CFKS23} proved EFX allocations exist when the agent's valuations are captured by a graph. 
    
    Another line of research has aimed to achieve multiplicative approximations of EFX. \citet{PR20} showed the existence of $1/2$-EFX allocations for subadditive valuations (subsequently derived in poly-time by \citet{CCLW19}), and \citet{AMN20} proved the existence of $1/\phi\approx 0.618$-EFX allocations for additive valuations. For more restrictive valuation functions, \citet{MS23} and \citet{ARS24} proved the existence of $2/3$-EFX allocations. 
    \citet{BKP24} achieved improved EFX approximations for restricted settings.

    For more general settings, a lot of work has focused on relaxations of EFX. One such relaxation has focused on ``partial allocations,'' donating some of the goods to charity and achieving EFX with the rest. This was first studied by \citet{CGH19} who showed the existence of a partial allocation that satisfies EFX and its Nash social welfare is half of the maximum possible. \citet{CKMS20} proved that EFX allocations exist for $n$ agents if up to $n-1$ goods can be donated; moreover, no agent envies the donated bundle. \citet{BCFF21} improved this number to $n-2$ goods with nice cancelable valuations, and \citet{M24} extended this to monotone valuations.
    \citet{BCFF21} further showed that an EFX allocation exists for $4$ agents with at most one donated good.   
    The number of donated goods was subsequently improved at the expense of achieving $(1-\varepsilon)$-EFX, instead of exact EFX \citep{CGMMM21, ACGMM22, BBK22, CSJS23}.      
    
    The existence of (approximate) EFX has also been studied along with Nash social welfare (NSW) guarantees. \citet{ABRHV21} showed that for additive bi-valued preferences, maximizing the NSW is always EFX. The tradeoff between EFX and NSW was recently studied by \citet{FMP24}.

    Another well-studied relaxation of EFX is {\em  envy freeness up to one good} (EF1), introduced by \citet{B10}. An allocation is EF1 if for every pair of agents $i$ and $j$, there exists some good in $j$'s bundle whose removal would ensure that $i$ does not envy $j$. \citet{LMMS04} showed that EF1 allocations always exist, even for general monotone valuations, and can be computed efficiently.
    There have been several other relaxations for EFX that are stronger than EF1, like EFL \cite{BBMN18}, EFR \cite{FHLSY21}, EEFX and MXS \cite{CGRSV22}, and their combinations (e.g., MXS and EFL \cite{AG24})).

    For the case of chore allocation, recently \citet{CS24} showed that the existence of EFX is not guaranteed even for instances with just 3 agents and 6 chores. The non-existence of EFX was also known for mixed goods and chores (mixed manna), where the valuations are not monotone \cite{BBBGKKKM24, HSVX23}.
    
    A broader overview of discrete fair division can be found in a recent survey by \citet{AABRLMVW23}.

    \section{Preliminaries}\label{prelims}
    An instance of discrete fair division is a tuple $\langle N,M,V\rangle$, where
    $N=[n]=\{1,2,\ldots,n\}$ is a set of agents, $M$ is a set of $m$ indivisible \itms, and $V= (v_1, v_2, \ldots, v_n)$ is a profile of valuation functions, where $v_i: 2^M \to \mathbb{R}_{\geq 0}$ for each agent $i\in N$ determines $i$'s value for each subset of goods. Whenever $v_i(S) < v_i(T)$, agent $i$ strictly prefers set $T$ to set $S$, and we denote this preference by $S \lowerval{i} T$; similarly, we use $S\lowereqval{i} T$ to denote weak preference, meaning that $v_i(S)\leq v_i(T)$. For notational simplicity, 
    we sometimes use $v(g)$ to denote $v(\{g\})$, i.e., the value for a single good, $S \cup g$ to denote $S \cup \{g\}$, and $S\setminus g$ 
    to denote $S \setminus \{g\}$. 
    Moreover, when we perform a union over disjoint sets, 
    we sometimes use $\sqcup$ instead of $\cup$ to emphasize their disjointness.

    \paragraph{Types of valuation functions.} 
    We consider valuation functions that are monotone, i.e., for any $S\subseteq T\subseteq M$, $v(S) \leq v(T)$. 
    A valuation function $v$ is \emph{non-degenerate} if $v (S) \ne v (T)$ for any two different bundles $S, T$.
    For convenience, throughout the paper, we assume that the valuation function of agent 1
    is non-degenerate (\cref{non degenerate} shows that
    this is without loss of generality). One of the most well-studied classes of valuation functions is \emph{additive}: a valuation function $v$ is additive if the value for any bundle $S\subseteq M$ is equal to the sum of the values of its goods, i.e., $v(S)=\sum_{g\in S} v(g)$.
    In this paper, we focus on the more general class of \emph{cancelable} valuations: a valuation function $v$ is \emph{cancelable}~\citep{BCFF21} 
    if for any two bundles $S, T \subset M$ 
    and any \itm \ $g \in M \setminus (S \cup T)$,
    if $v (S \cup g) > v (T \cup g)$, then $v (S) > v (T)$, i.e., removing the same good from two different bundles would not change the relative preference between the two. It is easy to verify
    that for any cancelable valuation $v$ and bundles $S, T, R$ such that $R \subseteq M \setminus (S \cup T)$ we have the following two properties, which we heavily use throughout the paper: 
    \begin{align*}
        v(S \cup R) > v(T\cup R)  &\Rightarrow v(S) > v(T) ~~~\text{and}\\
        v(S\cup R) \leq v(T\cup R) &\Leftarrow  v(S) \leq v(T).
    \end{align*}

    A class of valuations that is even more general than cancelable valuations\footnote{The fact that every cancelable valuation is MMS-feasible is shown in Lemma 2 of \citet{ACGMM22}. They claim that ``nice'' cancelable valuations are MMS-feasible, but their definition of nice cancelable coincides with our definition of a cancelable function, which is also the one originally used by \citet{BCFF21}.} is \emph{MMS-feasible} valuations: a valuation function $v$ is MMS-feasible if for every bundle $S \subseteq M$ and any two bipartitions
    of $S$, $\Y = (\ya, \yb)$ and $\yp = (\yap,\ybp)$, we have $\max(v(\ya), v(\yb)) \geq \min(v(\yap), v(\ybp))$.

    \paragraph{Sets of bundles and notation.}
    Throughout the paper, we use $\X= (X_1, X_2, \dots, X_k)$ and variants such as $\X'$ and $\tilde{\X}$ to denote a partition of all goods in $M$ into $k$ (pairwise disjoint) bundles, i.e. $\bigcup_{i\in[k]}X_i=M$. If the union of the bundles does not need to equal the set of all goods, we instead use $\Y = (Y_1, Y_2, \dots, Y_k)$. We let $\argmin_i(\Y)\in \arg\min_{j\in [k]} v_i(Y_j)$ be a least valuable bundle in $\Y$ from agent $i$'s perspective, 
    and $\argmax_i(\mathbf{Y})\in \arg\max_{j\in [k]} v_i(Y_j)$ to be a most valuable bundle. 
    In the extreme case where $v_i(Y_j)$ is the same for all $j\in [k]$ (and $k\geq 2$), we assign different bundles to $\argmin_i(\mathbf{Y})$ and $\argmax_i(\mathbf{Y})$, so it always holds that $\argmin_i(\mathbf{Y})\neq \argmax_i(\mathbf{Y})$. 
    For simplicity we write $\argmin_i(Y_1, Y_2, \dots , Y_k)$ instead of $\argmin_i((Y_1, Y_2, \dots , Y_k))$, and similarly for $\argmax_i$.
    
    Given an ordering $g_1\lowereqval{i} g_2 \lowereqval{i}\ldots \lowereqval{i} g_{|S|}$ of the goods of some bundle $S$ induced by the preferences of some agent $i$, we use $\lst{i}{k}{S}$ for some natural number $k\leq |S|$ to denote the $k$-th least valued \itm\ in bundle $S$ from agent $i$'s perspective, i.e., $\lst{i}{k}{S} = g_k$. Also, we let $\hst{i}{k}{S}$ be the $k$-th highest value \itm\ in $S$ from agent $i$'s perspective, i.e., $\hst{i}{k}{S}=g_{|S|+1-k}$. In the extreme case that $v_i(g)$ is the same for all $g\in S$ (and $|S|\geq 2$), we assign different goods  to $\lstitem{i}{1}{S}$ and $\hstitem{i}{1}{S}$, so we always have $\lstitem{i}{1}{S} \neq \hstitem{i}{1}{S}$.

\paragraph{\{EF, EFX, EF2X\}-envy.}  
    We say that an agent $i$ envies some bundle $T$ \with\ bundle $S$, if $S\lowerval{i} T$. 
    We say that agent $i$ \efxenvies\ $T$ \with\ $S$ 
    if there exists some $g\in T$ such that $S\lowerval{i} T\setminus g$, and we denote this by $S\efxenvyll{i} T$.
    We further write $S \doesnotefxenvyll{i} T $ if agent $i$ does not \efxenvy\  $T$ \with\ $S$. 
    We say that agent $i$ \eftxenvies \ $T$ \with\ $S$ 
    if there exist two distinct \itms\ $g_1,g_2 \in T$ 
    such that $S\lowerval{i} T \setminus \{g_1,g_2\}$.
    For simplicity, we use $(Y_1,...,Y_p) \doesnotefxenvyll{i} (Y'_1,...,Y'_q)$ 
    if agent $i$ does not \efxenvy\ any of the bundles $Y'_1,...,Y'_q$ \with\ 
    any of the bundles $Y_1,...,Y_p$.

    \paragraph{\{EFX, EF2X\}-feasibility.} 
    Given a set of disjoint bundles of \itms\ $\Y=(Y_1,...,Y_k)$, 
    we say that bundle $Y_j$ is \efxf\ for agent $i$ (or for valuation $v_i$) in $\Y$ 
    if agent $i$ does not \efxenvy\ any other bundle in $\Y$ \with\ $Y_j$.
    Similarly, we say that $Y_j$ is \eftxf\ for agent $i$ (or for valuation $v_i$) in $\Y$
    if agent $i$ does not \eftxenvy\ any bundle in $\Y$ \with\ $Y_j$. 
    We say that partition $\X$ is \efxf\ (respectively \eftxf) 
    if $k=n$ 
    and there is a way to assign each bundle in $\X$ to a distinct agent in $N$, 
    such that each bundle is \efxf\ (respectively \eftxf) for the agent it is assigned to in $\X$.

\paragraph{\{EFX, EF2X\}-best sets.} 
    Given a bundle $S\subseteq M$ and some natural number $k\leq |S|$, we use $G_k(S)= \{T\subseteq S : |T| = k\}$ to denote the set of subsets of $S$ of size $k$. Then, if $T^*\in \arg\max_{T \in G_k(S)}v_i (S\setminus T)$, we refer to $S\setminus T^*$, i.e., the best subset of $S$ that one can get after removing $k$ of its goods, as the \bestkremainer\ of $S$ w.r.t.\ agent $i$'s valuation. We refer to $w^k_i(S)=v_i (S\setminus T^*)$ as the {\em \bestkvalue} of a set $S$ for agent $i$; if $k>|S|$, we define the \bestkvalue\ of a set $S$ for any agent $i$ to be $0$. 
    We say that bundle $X_j$ is {\em EFX-best} and {\em EF2X-best} for agent $i$ in $\X$, if it has the maximum best-$1$-value and best-$2$-value, respectively, among all bundles in $\X$. 
    We write $\text{\efxfset}_i(\X)$, $\text{\eftxfset}_i(\X)$, and $\text{\bstfset}_i(\X)$
    to denote the set of \efxf, \eftxf, and \bstf\  bundles for agent $i$
    in partition $\X$, respectively.

\paragraph{The Plaut-Roughgarden (PR) algorithm.}
    This local search algorithm, dubbed the \emph{PR algorithm} by \citet{ACGMM22}, takes as input a set of bundles $\Y=(Y_1, Y_2,..., Y_k)$ and a valuation function $v$, and it returns a set of bundles $\yp=(Y'_1, Y'_2,..., Y'_k)$ that is \efxf\ for valuation $v$ in $\yp$ (i.e., all of its bundles are \efxf\ in $\yp$). If $Y_j$ is the least valuable bundle in $\Y$ w.r.t.\ valuation $v$, the algorithm checks whether $Y_j$ is \efxf\ w.r.t.\ $v$ in $\Y$. If it is, the algorithm terminates and returns $\Y$. On the other hand, if $Y_j$ is not \efxf\ in $\Y$, then there must exist another bundle $Y_i$ such that $v(Y_j)< v(Y_i \setminus g)$ for some good $g\in Y_i$. The PR algorithm then removes $g$ from $Y_i$, it adds it to $Y_j$, and it repeats the same sequence of steps until it reaches an \efxf\ set of bundles.

    Each time the PR algorithm moves a good, the minimum value over all the bundles, i.e., $\min_{Y\in \Y}v(Y)$, weakly increases, and if $v$ is non-degenerate, it strictly increases.
\begin{observation}
\label{PR}
    Let $\yp=(Y'_1,...,Y'_k)$ be the set of bundles returned if we run the PR algorithm with input $\Y=(Y_1,...,Y_n)$ and valuation function $v$. Then:
    \begin{equation*}
    \min(v(Y_1),...,v(Y_k)) \leq \min(v(Y'_1),...,v(Y'_k)).
    \end{equation*}
    Moreover, if there exists a bundle $Y_j$ that is not \efxf\ in $\Y$ w.r.t.\ valuation $v$ and $v$ is non-degenerate:
        \begin{equation*}
          \min\left(v(Y_1),...,v(Y_k)\right) < \min(v(Y'_1),...,v(Y'_k)).
        \end{equation*}
\end{observation}

\section{Computing \eftx \ Allocations for Four Agents}
\label{sec:4agents}
    
    Our main result (\cref{thm:main_result}) proves the existence of an EF2X allocation for instances involving an arbitrary number of goods and four agents with cancelable valuations; in fact, one of them can have an arbitrary monotone valuation. 
    
\begin{theorem}\label{thm:main_result}
    For every instance involving four agents with cancelable valuation functions and any number of goods, there exists an EF2X allocation, and we can compute one in pseudo-polynomial time. 
\end{theorem}

    As the theorem statement suggests, instead of a purely existential argument, we provide a constructive proof using an algorithm that computes an EF2X allocation in pseudo-polynomial time. The algorithm starts with some arbitrary initial partition of the goods into four bundles, and it gradually manipulates this until it reaches a partition that is \eftxf. Given the complexity of the problem, the ways in which the algorithm gradually transforms a partition aiming to achieve EF2X-feasibility are quite non-trivial and technical. Therefore, to be able to clearly present this algorithm, we define a sequence of stages (formally defined later on), each of which determines a list of structural properties that a partition in that stage needs to satisfy. Given these stages, we can provide a high-level description of this complicated algorithm using the simple diagram of \cref{fig:algorithm-flow}.

    Our algorithm follows a sequence of steps, each of which takes as input some partition $\X=(\xa, \xb, \xc, \xd)$ and returns a partition $\xt=(\xat, \xbt, \xct, \xdt)$. The way in which each step transforms $\X$ into $\xt$ depends on the stage that $\X$ is in, i.e., on the structural properties that $\X$ satisfies. Using these properties, our algorithm carefully re-arranges the goods across bundles and our proof shows that the resulting partition, $\xt$, satisfies all of the corresponding properties of a new stage. 
    The diagram of \cref{fig:algorithm-flow} illustrates the high-level structure of the algorithm by providing a node for each stage that a partition $\X$ may be in and a directed edge between two stages for each step that transforms a partition $\X$ to $\xt$. E.g., the edge from \stageTwoA\ to \stageTwoNormal\ corresponds to a transformation that takes as input a partition $\X$ in \stageTwoA\ and returns a partition $\xt$ in \stageTwoNormal\ (see \cref{ToNormal}). 
    
    The algorithm starts with an arbitrary initial partition, and it uses the PR algorithm to get a partition whose bundles are all \efxf\ for agent $1$. It then turns this into an \good\ partition (the first stage in~\cref{fig:algorithm-flow}). From that point onward, at the end of every step (i.e., after using each edge of~\cref{fig:algorithm-flow}), every bundle in the resulting partition $\xt$ will have at least one agent for whom it is \efxf\ in $\xt$. If we can match the agents to distinct bundles that are \eftxf\ for them, then the algorithm can terminate. If not, then we essentially identify an ``under-demanded'' bundle (one that is \efxf\ for agent $1$, but not \eftxf\ for any other agent) and some ``over-demanded'' bundles. Each step then carefully re-distributes goods across these bundles to make the under-demanded bundle more appealing.

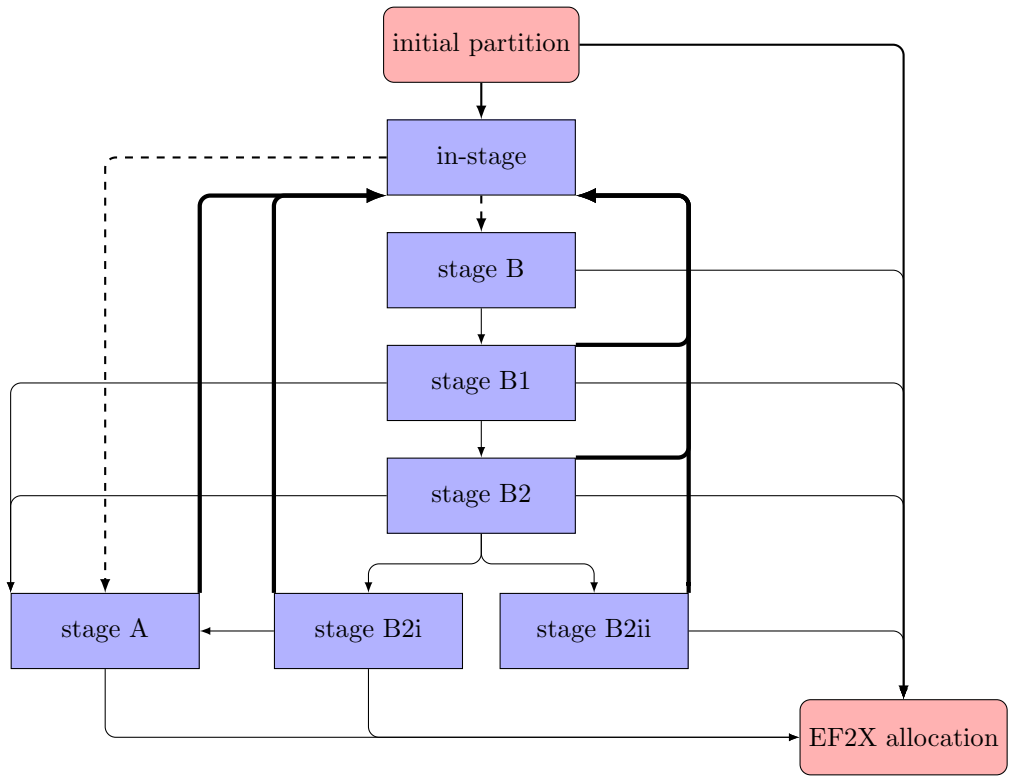
\begin{figure}
    \centering
    \begin{tikzpicture}[node distance=1.5cm and 2cm]
<TikZ code>
\node (start) [startstop] {initial partition};
\node (good part) [process, below of=start] {\good};
\node (stage 2) [process, below of=good part] {\stageTwo};
\node (stage 2A) [process, below of=stage 2]{\stageTwoA};
\node (normal) [process, below of=stage 2A]{\stageTwoNormal};
\node (stage 2C) [process, below of=normal, xshift=-1.5cm, yshift=-0.3 cm]{\stageTwoC};
\node (perfect) [process, below of=normal, xshift=1.5cm, yshift=-0.3 cm] {\stageTwoPerfect};
\node (stage 1) [process, left of=stage 2C, node distance=3.5 cm] {\stageOne};
\node (EF2X) [startstop, below right of=perfect, node distance=2cm, xshift=2.7cm] {\eftx\ allocation};

\path[line, thick, rounded corners] (start) -- (good part);
\path[line, thick, rounded corners] (start) -| (EF2X);

\path[line, thick, dashed, rounded corners] (good part.west) -| (stage 1);
\path[line, thick, dashed, rounded corners] (good part) --  (stage 2);

\path[line, ultra thick, rounded corners] (stage 1.north east) |-  (good part.south west);
\path[line, rounded corners] (stage 1) |- (EF2X);

\path[line, rounded corners] (stage 2) --  (stage 2A);
\path[line, rounded corners] (stage 2.east) -|  (EF2X);

\path[line, rounded corners] (stage 2A) -- (normal);
\path[line, rounded corners] (stage 2A.east) -| (EF2X);
\path[line, rounded corners] (stage 2A.west) -| (stage 1.north west);
\path[line, ultra thick, rounded corners] (stage 2A.north east) -| ($(perfect.north east) + (0,4)$) |- (good part.south east);

\path[line, rounded corners] (normal.east) -| (EF2X);
\path[line, rounded corners] (normal.west) -| (stage 1.north west);
\path[line, ultra thick, rounded corners] (normal.north east) -| ($(perfect.north east) + (0,4)$) |- (good part.south east);
\path[line, rounded corners] (normal.south) |- ($(normal.south) + (0.5,-0.4)$) -| (perfect.north);
\path[line, rounded corners] (normal.south) |- ($(normal.south) + (-0.5,-0.4)$) -| (stage 2C.north);

\path[line, rounded corners] (stage 2C.south) |- (EF2X);
\path[line, rounded corners] (stage 2C.west) -- (stage 1);
\path[line, ultra thick, rounded corners] (stage 2C.north west) |-  (good part.south west);

\path[line, rounded corners] (perfect.east) -| (EF2X);
\path[line, ultra thick, rounded corners] (perfect.north east) -| ($(perfect.north east) + (0,4)$) |- (good part.south east);
\end{tikzpicture}
    \caption{Dashed edges correspond to transitions where the partition remains the same. Non-dashed edges correspond to transitions where the partition may change. The non-dashed edges are bold whenever the potential strictly increases during the transition and non-bold whenever the potential weakly increases.}
    \label{fig:algorithm-flow}
\end{figure}

    \paragraph{Potential function and running time.} To ensure that the algorithm is ``making progress'' toward its goal of reaching an EF2X allocation, we use a potential function that weakly increases after each step.
    Specifically, we set aside one agent (Agent 1), who plays a special role throughout the algorithm, and we use that agent's value for the first bundle of the partition as the potential of that partition.
    \begin{definition}[potential function]
    The {\em potential} of a partition $\X$ is $\phi(\X) = v_1(X_1)$.
    \end{definition}
    Note that the bold edges in \cref{fig:algorithm-flow} (all edges going upward) correspond to steps of the algorithm that guarantee a strict increase of the potential. All other edges ensure that the potential does not drop. Using this observation, combined with the fact that the number of potential increases can be at most pseudo-polynomial, and the fact that each step takes no more than pseudo-polynomial time, we can conclude that the overall running time of this algorithm is pseudo-polynomial (a more detailed analysis can be found in \cref{running_time_4agents}).

    Note that throughout the algorithm agent $1$, who has a nondegenerate valuation, is fixed. On the other hand, the remaining three agents, $\{2,3,4\}$, are not fixed, and we instead use indices $i,j,u$ to refer to them and their preferences.

\subsection*{Formal Definitions of the Algorithm's Stages.}
    We now provide the formal definitions of the stages.

\begin{definition}[\stageOne]
    We say partition $\X$ is in \stageOne\ 
    if \areefxffor{\xa,\xb,\xc}{1}{\X} with
    $\xa = \argmin_1(X_1, X_2, X_3)$, 
    and at least one of the following holds:
    
    i) $\xd$ is \bstf \ for at least one agent $j \ne 1$ and
     $\xd  \setminus \hj  \greatereqval{j}  X_1 \cup \hj$, 
    
    ii) $\xd$ is \bstf \ for at least two distinct agents $i,j \ne 1$. 
\end{definition}

\begin{definition}[\stageTwo]
    We say partition $\X$ is in \stageTwo\ 
    if \areefxffor{\xa,\xb}{1}{\X}, 
    $\xa= \argmin_1(X_1, X_2)$,
    \isefxf{\xc}{i}{\X} for some agent $i \ne 1$, and \isefxf{\xd}{j}{\X} for some agent $j\notin \{1,i\}$.
\end{definition}

\begin{definition}[\goodpart]
    We say a partition is \good\ if it is either in \stageOne\ or in \stageTwo.
\end{definition}

\begin{definition}[\stageTwoA]
    We say partition $\X$ is in \stageTwoA\
    if it is in \stageTwo, \isbstffor{\xd}{j}{\X}, and
    \begin{center}
        $\xd \setminus \hj \greaterval{j} (\xa, \xb)$.
    \end{center}
    \end{definition}

\begin{definition}[\normalized\ sets]
\label{normalized sets def}
    We say that two sets $(S,T)$ are \normalized\ w.r.t.\ the agents $i$ and $j$ 
    if there exists no $g \in S, h \in T$ such that
    $g\lowereqval{i} h$, $h\lowereqval{j} g$, 
    and at least one of these two preferences is strict.
\end{definition}
     
\begin{definition}[\stageTwoNormal]
    \label{normal stage 2 def}
    We say partition $\X$ is in \stageTwoNormal\ 
    if it is in \stageTwoA \ and
    $(\xc, \xd)$ are \normalized\  w.r.t.\ the agents $i$ and $j$ that satisfy the conditions of \stageTwo.
\end{definition}

\begin{definition}[\stageTwoC]
    We say partition $\X$ is in \stageTwoC\ if it is in \stageTwoNormal \ and it also satisfies all the following conditions:
    \begin{align*}
        \xd \setminus \{\hi, \Hi\}  &\greaterval{i} \xa , \xb \\
        \xd \setminus \{\hj, \Hj\}  &\greaterval{j} \xa , \xb \\
        \argmax_i (\xa, \xb) &\greatereqval{i} \xc \setminus \gi \\
        \xc \cup \hj &\efxenvygg{j} \xd \setminus \hj 
    \end{align*}
\end{definition}

\begin{definition}[\stageTwoPerfect]
    \label{perfect def}
    We say partition $\X$ is in \stageTwoPerfect, if 
    it is in \stageTwoNormal\
    and
    \begin{equation*}
        \xc \setminus \gi \greaterval{i} (\xa, \xb).
    \end{equation*}
\end{definition}

\subsubsection*{Outline of the Section}

    In \cref{techniques}, we provide some of the high-level ideas used by our algorithm and new concepts introduced in the paper. 
    In  \cref{cancelabel valuations properties},  
    we provide some lemmas and observations that we will be using throughout the paper.
    In \cref{{stage A or B section}}, we provide the way we reallocate goods given either a \stageOne\ or a \stageTwo\ partition.
    In \cref{swap optimization section}, we define the \normalizationprocess\ process and its properties. In \cref{stage 2A section}, we provide the steps we follow given an partition in \stageTwoA.  In \cref{stage 2B section}, we provide the steps we follow given a \stageTwoNormal\ partition, and in \cref{stage 2c section}, we provide the steps for \stageTwoC \ partitions.
    Finally, in \cref{stage 2d section}, we conclude with the steps followed given a  \stageTwoPerfect\ partition.

\subsection{Our Techniques}
\label{techniques}

    We now present an overview of some of the high-level ideas of our algorithm and new concepts introduced in the paper.

    \paragraph{EFkX-Best-Bundle.}
    One concept that we introduce is the \bstkf\ bundle in a set of bundles $\Y$ for some agent $i$ (see \cref{prelims} for a definition). 
    This satisfies the very useful property that any bundle that is \bstkf\ (and, hence, also \efkxf) for some agent $i$, remains \efkxf\ even if we remove $i$'s $k$ least valued \itms\ from it (and keep those \itms\ unallocated). This holds since the remaining bundle is, by definition, at least as valuable for $i$ as any other bundle $T$, after the removal of any $k$ \itms\ from $T$. This is a property that we leverage when, during some reallocation process, we remove goods from some bundle, temporarily leading to a {\em partial} partition.
    
    To extract useful conditions for {\em full} partitions, we combine the above idea with the following observation: Given a partition, let $S$ be the \bsttwof\  bundle for
    some agent $i$ and $T$ be another bundle that is not \eftxf\ for $i$. Then, if $i$ received $T$, she would EF2X-envy $S$; even if we removed $i$'s two least valuable \itms\ from $S$, she would still prefer $S$ to $T$. However, the trick here is to instead remove just the {\em one} least valuable \itm\ of $S$ and add it to $T$. This gives a (full) partition where agent $i$ still prefers the new bundle $S$ to the new bundle $T$ (this is due to cancelable valuations). Moreover, since $S$ was \bsttwof\ for $i$ before, it remains at least \eftxf\ for $i$ after this change. 
    This way, some ``underdemanded'' bundle has gained a new good, and the one that was \bsttwof\ for $i$ is still at least \eftxf\ for her, even though it lost a good. Such re-allocations allow us to make progress toward the final allocation.
    
    Although this procedure will eventually increase the number of bundles that are \eftxf\ for an agent, we cannot use it by itself: it may compromise EF2X-feasibility for other agents, who may end up EF2X-envying the bundle that gained a good. To avoid this, we require conditions for other agents (one example is achieved by the \normalizationprocess\ that we discuss next) and possibly some case analysis.

\paragraph{\NormalizationProcess.}
    We also introduce a procedure we call {\em \normalizationprocess}, in which two bundles, each associated with a different agent, swap items in such a way that those agents are more satisfied with their new bundles. 
   
    \NormalizationProcess\ is useful to overcome some issues discussed at the end of the EFkX-Best-Bundle paragraph.
    Specifically, consider a partition of the goods such that three bundles, $S,T$, and $Q$, satisfy that $S$ is \bsttwof\ for agent $i$, $T$ is \bsttwof\ for agent $j$,
    and $Q$ is not \eftxf\ for either of them. Given such conditions, the EFkX-Best-Bundle paragraph suggested moving to $Q$ either $i$'s least valued \itm\ from $S$ or $j$'s least valued \itm\ from $T$. To ensure that this does not cause any envy, we can first perform \normalizationprocess\ between $S,T$ w.r.t.\ $i,j$, which maintains these initial conditions (by Lemmas~\ref{Normalization best-k-feasibility}~and~\ref{best props}). If $g$ is $i$'s least valuable good in $S$ and $h$ is $j$'s least valuable good in $T$ after the swap-optimization, we move to $Q$ $i$'s least valuable \itm\ in $\{g,h\}$. This guarantees that neither $i$ nor $j$ will envy the resulting bundle $Q$. To see this, suppose $i$ weakly prefers $h$ to $g$. Then, $j$ also weakly prefers $h$ to $g$, since $S$ and $T$ were swap-optimized. As we argued in the EFkX-Best-Bundle paragraph, $i$ will not envy the resulting bundle $Q$ after adding $g$ to it. Similarly, agent $j$ would not envy $Q$ even if we were to move $h$ from $T$ to $Q$, so she will not envy $Q$ if we add $g$ to it, which she values less, and keep $T$ unchanged. Similar arguments apply if the moved good is $h$.

\paragraph{Maintaining EFX-feasibility for 2 agents and more goods.}
    %The following fact also plays a crucial role in our algorithm.
    Suppose $i,j$ are two agents with MMS-feasible valuations, %(a generalization of cancelable valuations),
    and  
    $\Y=(\ya,\yb)$ are two disjoint bundles such that $\ya$ and $\yb$ are \efxf\ for agents $i,j$, respectively, in $\Y$. 
    Suppose now that we want to introduce a set $\unassignedItms$ of additional \itms\ and re-allocate, so that the above conditions are preserved.
    We will show that (in proof of \cref{ToGoodPart3})
    it is possible to rearrange all those \itms\ and derive a set of two disjoint bundles $\yp=(\yap,\ybp)$ (where $\yap\sqcup\ybp= \ya \sqcup \yb \sqcup \unassignedItms$) such that 
    $\yap$ and $\ybp$ are \efxf\ for agents $i,j$, respectively, in $\yp$, 
    and moreover, both agents are weakly more satisfied than before, i.e., $\yap \greatereqval{i} \ya$ and $\ybp \greatereqval{j} \yb$.

    This tool is useful in cases where two agents are \efxf\ with one distinct bundle each, and in our attempt to satisfy another agent (agent $1$), we need to remove some goods from another bundle and then add them %\st{add some goods} 
    to the bundles associated with these two agents without breaking the condition that those bundles are \efxf\ for them.

% \paragraph{Maintaining EFX-feasibility for 2 agents and more goods.}
% %The following fact also plays a crucial role in our algorithm.
% Suppose that $i,j$ are two agents with MMS-feasible valuations (a generalization of cancelable valuations), and  
%     $\Y=(\ya,\yb)$ are two disjoint bundles such that $\ya$ and $\yb$ are \efxf\ for agents $i,j$, respectively, in $\Y$. 
%     Suppose now that we want to include a set $\unassignedItms$ of extra \itms\ to those bundles, such that the above conditions are preserved.
%    We will show that (in proof of \cref{ToGoodPart3}) it is possible to rearrange all those \itms\ and derive a set of two disjoint bundles $\yp=(\yap,\ybp)$ (where $\yap\sqcup\ybp= \ya \sqcup \yb \sqcup \unassignedItms$) such that 
%     $\yap$ and $\ybp$ are \efxf\ for agents $i,j$, respectively, in $\yp$, 
%     and moreover, both agents are weakly more satisfied than before, i.e.\, $\yap \greatereqval{i} \ya$ and $\ybp \greatereqval{j} \yb$.
    
%     This tool is useful in cases where two agents are \efxf\ with one distinct bundle each, and in our attempt to satisfy another agent (agent $1$) we need to remove some goods from another bundle and then add them 
%     to the bundles associated with the other two agents without breaking the condition that those bundles are \efxf\ for them. 

\subsection{Some Useful Observations and Lemmas}

Before diving into the main technical sections of the paper, we dedicate this section to proving some facts that we use repeatedly later on. We first provide some useful observations regarding the class of cancelable valuations. We then focus on properties related to \bstkf\ and \efkxf\ bundles, and we conclude with some lemmas and observations regarding some fundamental transformations of partitions to achieve desired properties, which we often use as building blocks for our procedures later on.

\subsubsection*{Observations regarding cancelable valuations}
\label{cancelabel valuations properties}

\begin{observation}
\label{cancelable prop 1}
    If agent $i$ has a cancelable valuation function, then for any two disjoint bundles $S, T$ and any \itms\ $g_1, g_2 \in S$ such that
    $g_1\lowereqval{i} g_2$ and $S \setminus \{g_1,g_2\}  \greatereqval{i} T$ we have
\begin{center}
    $S \setminus g_1 \greatereqval{i} T \cup g_2$ ~~and~~ $S \setminus g_1  \greatereqval{i} T \cup g_1$.
\end{center}
\end{observation}
\begin{proof}
    The first condition is due to cancelability, by adding $g_2$ to $S \setminus \{g_1,g_2\}$ and to $T$. The second condition  
    is due to the fact that $g_1\lowereqval{i} g_2$, 
    which implies that $T \cup g_2 \greatereqval{i} T \cup g_1$.
\end{proof}

\begin{observation}
\label{cancp}
    If agent $i$ has a cancelable valuation function and for some bundles $S,T,Q,R$ 
    we have $S\lowereqval{i} T$, $Q\lowereqval{i} R$, and 
    $(S\cup T)\cap(Q\cup R)=\emptyset$,
    then $S\cup Q \lowereqval{i}  T \cup R$.
\end{observation}

\begin{proof}
    Since $S\lowereqval{i} T$, and $Q$ is disjoint from $S$ and $T$, 
    by cancelability it holds that $S \cup Q \lowereqval{i} T \cup Q$. 
    Similarly, since $Q \lowereqval{i} R$, and $T$ is disjoint from $Q$ and $R$, 
    by cancelability it holds that $ T \cup Q \lowereqval{i} T \cup R$. By combining these two, we have that
    $S \cup Q \lowereqval{i} T \cup Q \lowereqval{i} T \cup R.$
\end{proof}

\begin{lemma}
\label{k-min with k-least}
    Given a bundle $S$ and an agent $i$ with a cancelable valuation function, 
    if $g_1,g_2,\ldots, g_k$ are the $k$ least valued \itms\ in $S$ w.r.t.\ agent $i$,
    then for any $k$ \itms\ $h_1,h_2,\ldots, h_k$ in $S$ we have 
    $S \setminus (g_1,g_2,\ldots,g_k) \greatereqval{i} S \setminus (h_1,h_2,\ldots,h_k)$. 
    Therefore, if $(g_1,\ldots, g_k)$ are the $k$ least valued \itms\ w.r.t.\ agent $i$'s valuation
    in $S$, then
    $w^k_i(S)= v_i(S\setminus (g_1,\ldots,g_k))$.
\end{lemma}
\begin{proof}
    Let $Q=\{g_1,g_2,\ldots,g_k\}, R = \{h_1,h_2,\ldots,h_k\}$, and
    w.l.o.g.\ assume that $Q\setminus R = \{g'_1,\ldots,g'_p\}$ 
    and $R\setminus Q = \{h'_1,\ldots,h'_p\}$
    such that 
    $g'_1\lowereqval{i} g'_2 \lowereqval{i} \ldots \lowereqval{i} g'_p$ and 
    $h'_1\lowereqval{i} h'_2 \lowereqval{i} \ldots \lowereqval{i} h'_p$, 
    then we have that for any $j \in [p]$, 
    $g'_j \lowereqval{i} h'_j$, 
    hence, by repeatedly using \cref{cancp},
    $Q\setminus R=\{g'_1,\ldots,g'_p\} \lowereqval{i} \{h'_1,\ldots,h'_p\}=R\setminus Q$. 
    Trivially, $S \setminus (Q\cup R) \lowereqval{i} S \setminus (Q \cup R)$, 
    so by \cref{cancp}, we get
    $S\setminus R= \big(S \setminus (Q\cup R)\big) \cup (Q \setminus R) \lowereqval{i} 
    \big(S \setminus (Q\cup R)\big)\cup(R \setminus Q)= S \setminus Q$, 
    therefore $S\setminus R\lowereqval{i} S\setminus Q$.
\end{proof}

\subsubsection*{Observations regarding \efkxf\ and \bstkf\ bundles}

\begin{observation}   
\label{EFX prop:1}
    Agent $i$ \efxenvies\ set $T$ \with\ set $S$ if and only if $v_i(S)< w^1_i(T)$.
    Suppose $\X$ is a partition. If for some agent $i$, \isnotefxf{X_j}{i}{\X} and 
    $X_a \lowereqval{i} X_j$, then \isnotefxffor{X_a}{i}{\X}.  
    Also if \isefxffor{X_j}{i}{\X} 
    and $X_j \lowereqval{i} X_a$, then \isefxffor{X_a}{i}{\X}. 
    All the above also hold for \eftx-feasibility.
\end{observation}

\begin{observation}
\label{best props}
    Suppose $\X$ is a partition.
    Then for every agent $i$ there exists an \bstf\ and an \bsttwof\ bundle. 
    If \isbstffor{X_j}{i}{\X}, then \isefxffor{X_j}{i}{\X}. 
    Also if \isbsttwoffor{X_j}{i}{\X}, 
    then \iseftxffor{X_j}{i}{\X}. 
    If \isbstffor{X_a}{i}{\X}, then \isefxffor{X_j}{i}{\X} if and only if
    $X_j \greatereqval{i} X_a \setminus \lst{i}{1}{X_a}$.
    If \isbsttwoffor{X_a}{i}{\X}, then \iseftxffor{X_j}{i}{\X} if and only if
    $X_j \greatereqval{i} X_a \setminus \{\lst{i}{1}{X_a},\lst{i}{2}{X_a}\}$.
\end{observation}

\begin{proof}
    If $X_j \lowerval{i} X_a \setminus \lst{i}{1}{X_a}$, by definition, $X_j$ is not \efxf. 
    On the other hand, if $X_j \greatereqval{i} X_a \setminus \lst{i}{1}{X_a}$, 
    since the valuations are cancelable, by \cref{k-min with k-least}, we get 
    $w_i^1(X_a) = v_i(X_a \setminus \lst{i}{1}{X_a})$.
    Since $X_a$ is \bstf \ for agent $i$, 
    for every bundle $X_u$ in the partition, we get that:
    \begin{equation*}
        v_i(X_j) \geq w^1_i(X_a) \geq w^1_i(X_u) 
    \end{equation*}
    Similar arguments hold for \eftxf\ and \bsttwof \ bundles. 
\end{proof}

\begin{observation}
\label{best-feasible condition}
    Consider a partition $\X$ and suppose that for some agent $i$ and bundle $X_j$, it is \isbstffor{X_j}{i}{\X}.  
    Then agent $i$ does not \efxenvy\ any bundle $X_a \in \X$
    \with\ bundle $X_j\setminus \lst{i}{1}{X_j}$.
\end{observation}
\begin{proof}
    Since \isbstffor{X_j}{i}{\X}, for any $X_a \in \X$ it holds that   
    $v_i(X_j\setminus \lst{i}{1}{X_j}) = w^1_i(X_j) \geq w^1_i(X_a)$.
\end{proof}

\begin{observation}
\label{get item to bstf}
    Suppose that in partition $\X$, some agent $i$ does not \eftxenvy\ some other bundle $X_b$
    \with\ bundle $X_a$, and does not \efxenvy\ any other bundle in the partition \with\ $X_a$.
    Then, moving any subset  $\unassignedItms \ne \emptyset$ from $X_b$ 
	 to $X_a$ results in a partition $\xp$, for which \isbstffor{X'_a}{i}{\xp}.
\end{observation}

\begin{proof}
    Since $X'_a = X_a \cup \unassignedItms$, we get that 
    $w^1_i(X'_a)= \max_{x \in X'_a}  (X'_a \setminus x) \geq v_i(X_a)$,
	where the inequality comes from the fact that $\unassignedItms \ne \emptyset$.
    We also have that:
    \begin{center}
    $v_i(X_a) \geq w^2_i(X_b) =  \max_{x,y \in X_b} v_i(X_b \setminus (x,y)) 
    \geq \max_{x \in X'_b} v_i(X'_b \setminus x) = w^1_i(X'_b)$        
    \end{center}
    where the first inequality comes from the fact that agent $i$ does not \eftxenvy\ bundle $X_b$ \with\ bundle $X_a$, and the last inequality comes from the fact that 
    $X'_b = X_b \setminus \unassignedItms$ and $\unassignedItms \ne \emptyset$.
    Moreover, for any other bundle $X'_j\in \xp$ except $X'_a,X'_b$,
    we have that $X'_j = X_j$ and so, $v_i(X_a) \geq w^1_i(X_j)$. 
    Overall, since $\unassignedItms \ne \emptyset$, it holds that $w^1_i(X'_a) \geq v_i(X_a) \geq w^1_i(X'_j)$, for any $X'_j\in \xp$, and so, \isbstffor{X'_a}{i}{\xp}.
\end{proof}

\subsubsection*{Lemmas and observations for partition transformations}
\label{Reductions to good partition}

\begin{lemma}
\label{toStage1}
    Given any partition $\X$ in which every bundle is \efxf\ for agent $1$, we can construct a new partition $\xt$ that is either \eftxf\ or \good\
    with potential $\phi(\xt)\geq\minfour$.
\end{lemma}
\begin{proof}
We first ask every agent other than 1 to request their \bstf \ bundle. If each of these agents requests a distinct \bstf\ bundle, then we can get an \eftxf\ partition (in fact, we even get a \efxf\ partition) by giving each of them their \bstf\ bundle  (which is also \efxf, by \cref{best props}) and giving agent 1 the remaining bundle, which is \efxf\ for them (like all bundles in $\X$). Otherwise, there exist two agents $i,j \ne 1$ that have the same \bstf\ bundle. We can rename the bundles so that i) this ``shared'' \bstf\ bundle is $\xd$ and ii) the least valuable bundle among the other three bundles w.r.t.\ agent 1's valuation is $X_1$. Then, the resulting partition $\xt$ is in \stageOneii\ and, because it has the same bundles as $\X$, we get $\phi(\xt)=v_1(\xat)=\minfour$.
\end{proof}

Using \cref{toStage1}, combined with the PR algorithm, we can now prove the following statement, which we will use as the first step of our algorithm.

\begin{lemma}
\label{ToGoodPart2}
Given any partition $\X$, we can construct a new partition $\xt$ that is either \eftxf\ or \good\ with $\phi(\xt)\geq\minfour$. Also, if at least one bundle in $\X$ is not \efxf\ for agent $1$ in $\X$, then $\phi(\xt)>\minfour$.
\end{lemma}
\begin{proof}
    We run the PR algorithm on $\X$ w.r.t.\ the valuation function of agent $1$, so all the bundles in the resulting partition $\xp$ are \efxf\ for agent $1$. Then, by using \cref{toStage1} on $\xp$, we get partition $\xt$ that is either \eftxf\ or \good\ with $\phi(\xt) \geq \minfourp$. By using \cref{PR}, we have that $\phi(\xt) \geq \minfourp \geq \minfour$. 
    Moreover, if there exists at least one bundle in $\X$ that is not \efxf\ for agent $1$ in $\X$, then \cref{PR} implies that $\minfourp > \minfour$, and therefore $\phi(\xt) > \minfour$.
\end{proof}

\begin{lemma}
\label{ToGoodPart0}
    Given any partition $\X$ such that \areefxffor{\xa,\xb,\xc}{1}{\X}, 
    and for some agent $j\ne 1$, \isefxf{\xd}{j}{\X},
    we can construct a new partition $\xt$ that is either \eftxf\ or \good\ with $\phi(\xt)\geq \minthree$.
\end{lemma}
\begin{proof}
    If  there are two agents for whom $\xd$ is \bstf\ in $\X$, 
    rename the minimum value bundle w.r.t.\ agent $1$'s valuation among 
    $(X_1,X_2,X_3)$  to $X_1$ and the other two bundles to $X_2,X_3$ accordingly, 
    then we get a partition $\xt$ that is in \stageOneii\ with $\phi(\xt)=\minthree$.

    Assume now that there are not two agents for whom $\xd$ is \bstf\ in $\X$. If there exists some agent for whom $\xd$ is \bstf\ in $\X$, name this agent $j$; otherwise, agent $j$ is as in the lemma's statement. In any case, \isefxffor{\xd}{j}{\X}, and for any agent other than $1,j$, $\xd$ is not \bstf. 
    This means that for some agent $i\notin \{1,j\}$, some \isbstffor{\X_c}{i}{\X}, with $\X_c\neq \xd$, and therefore it also holds that \isefxffor{\X_c}{i}{\X}.
    Rename bundle $X_c$ to $X_3$ and the two bundles 
    $\{X_1,X_2,X_3\}\setminus \{X_c\}$  to $(X_1,X_2)$ such that $X_1 \lowerval{1} X_2$.
    This new partition $\xt$ is in \stageTwo \ and its potential is 
    $\phi(\xt) = \minonet \geq \minthree$.
\end{proof}

\begin{lemma}
\label{ToGoodPart1}
    Given any partition $\X$ such that $\xd$ is \efxf\ for some agent $j \ne 1$,  we can construct a new partition $\xt$ that is either \eftxf\ or \good\ with $\phi(\xt)\geq \minthree$.
\end{lemma}

\begin{proof}
    Rearrange the bundles such that $\xa$ has the minimum value amongst $\partt$ w.r.t.\ agent $1$'s valuation. Then for $r \in \{2,3\} $, let $X'_r$ be a minimal subset 
    of $X_r$ such that $X'_r \greaterval{1} \xa$, and then let
    $\xp=(\xa,\xbp,\xcp,\xd\cup (\xb\setminus \xbp)\cup(\xc\setminus \xcp))$. 
    Then, we have  $\xdp \greatereqval{j} \xd$ because $\xd \subseteq \xdp$. 
    Since agent $j$ does not \efxenvy\ any bundle $X \in \{\xa,\xb,\xc\}$
    \with\ bundle $\xd$, 
    she does not \efxenvy\ any subset of  $X$
    \with\ bundle $\xdp$. 
    Hence \isefxffor{\xdp}{j}{\xp}.
    
    Hence, if the bundles $\parttp$ are \efxf\ for agent $1$ in $\xp$,
    by \cref{ToGoodPart0}, 
    we can construct a new partition $\xt$ that is either 
    \eftxf\ or \good\ with $\phi(\xt)\geq \minthreep = \minthree$.
    The equality holds because $\xap$ has the least value among  $(\xap, \xbp, \xcp)$
    w.r.t.\ agent $1$'s valuation, and equals to bundle $X_1$.
    
    So, assume otherwise. We have that $\parttp \doesnotefxenvyll{1} \parttp$. 
    Hence, agent $1$ \efxenvies\ $\xdp$ \with\ bundle $\xap$.
    Therefore $\xdp \greaterval{1} \xap = \argmin_1 \parttp$. 
    Hence, by \cref{ToGoodPart2},  
    we can construct a new partition $\xt$ 
    that is either \eftxf\ or \good\ with $\phi(\xt) > \minfourp=\minthree$.
\end{proof}

\begin{lemma}
\label{ToGoodPart4}
    Given any partition $\X$ such that for some agent $i\neq 1$ and $j \notin \{1,i\}$, 
    \isefxffor{\xc}{i}{\X}, \isefxffor{\xd}{j}{\X}, 
    $\xa \doesnotefxenvyll{1} \xb$, 
    and $\xa\lowerval{1} \xb$, we can construct a new partition $\xt$ that is either \eftxf\ or \good\ with $\phi(\xt)\geq \minone$.
\end{lemma}
\begin{proof}
    If \isefxffor{(\xa,\xb)}{1}{\X}, by the conditions of the lemma's statement, $\X$ would be in \stageTwo \ with 
    $\phi(\X) = \minone$. In a different case, since  $\xa\lowerval{1} \xb$, it should be that \arenotefxffor{\xa}{1}{\X} and so, $\xa \lowerval{1} \xc$ or $\xa \lowerval{1} \xd$. W.l.o.g.\ assume that
    $\xa \lowerval{1} \xc$. 
    Then, by \cref{ToGoodPart1}, 
    we can construct a new partition $\xt$ that is either \eftxf\ or \good\ 
    with $\phi(\xt)\geq \minthree= \minone$.
\end{proof}

\begin{lemma}
\label{ToGoodPart3}
    Suppose $\Y=(\ya,\yb,\yc,\yd)$ is a set of  disjoint bundles (the union of which may not be the set of all \itms\ $M$),
    such that for some agent $i\neq 1$, \isefxffor{\yc}{i}{\Y}, and for some agent $j \notin \{1,i\}$, \isefxffor{\yd}{j}{\Y}.
    Then, we can construct a partition $\xt$ of all the \itms\ in $M$
    that is either \eftxf\ or \good\ 
    with $\phi(\xt) \geq \mintwoy$.
\end{lemma}

\begin{proof}
    W.l.o.g.\ assume that $Y_1\lowerval{1} Y_2$ (otherwise swap the two). 
    Let $\unassignedItms_1$ be the set of \itms\ that are not used in $\Y$.  
    Let $\xbp$ be a minimal subset of $\yb$ such that $\xbp \greaterval{1} \ya$. 
    Therefore $\ya \doesnotefxenvyll{1} \xbp$.
    Let $\unassignedItms  = \unassignedItms_1 \cup (\yb \setminus \xbp)$.
    
    If $\yd \doesnotefxenvyll{j} \yc \cup \unassignedItms$,
    set $\xp = (\ya, \xbp , \yc \cup \unassignedItms, \yd)$.
    Then \isefxffor{\xdp=\yd}{j}{\xp}, 
    since we have that \isefxf{\yd}{j}{\X}, 
    so agent $j$ does not \efxenvy\ the bundles $\ya,\xbp$ \with\ bundle $\yd$.
    Also \isefxf{\xcp}{i}{\xp}, since \isefxf{\yc}{i}{\Y}, $\yc \subseteq \xcp$, 
    and we have not added more \itms\ to the other bundles.
    If $\yc \doesnotefxenvyll{i} \yd \cup \unassignedItms$, 
    set $\xp = (\ya, \xbp , \yc, \yd \cup \unassignedItms)$.
     Then \isefxf{\xdp}{j}{\xp} and \isefxf{\xcp}{i}{\xp} (with similar arguments). 
    
    Otherwise, we have that  $\yd \lowerval{j} \yc \cup \unassignedItms$, 
    and  $\yc \lowerval{i} \yd \cup \unassignedItms$. 
    Run the PR algorithm w.r.t.\ agent $i$'s valuation 
    on the $(\yc, \yd \cup \unassignedItms)$ to get $(\ycp,\ydp)$. 
    W.l.o.g.\ assume that $\ydp \greatereqval{j} \ycp$. 
    Now set $\xp = (\ya, \xbp , \ycp, \ydp)$.
    Note that, since agent $j$'s valuation is cancelable, it is MMS-feasible too. Therefore, 
    $\ydp = \argmax_j(\ycp,\ydp) \greatereqval{j} \argmin_j (\yd, \yc \cup \unassignedItms ) = \yd$. Since $\yd \doesnotefxenvyll{j} (\xap,\xbp)$, it holds that \isefxffor{\ydp}{j}{\xp}.
    Regarding agent $i$, $\ycp, \ydp$ are the output of the PR algorithm with agent $i$'s valuation, 
    $\ycp \doesnotefxenvyll{i} \ydp$, and since the valuations are MMS-feasible
    $\ycp \greatereqval{i} \argmin_i (\yc, \yd \cup \unassignedItms) = \yc$. Since $\yc \doesnotefxenvyll{i} (\xap,\xbp)$ it holds that 
    \isefxffor{\ycp}{i}{\xp}.
    
    Hence, in each case, we manage to construct a partition $\xp$  such that
    \isefxffor{\xcp}{i}{\xp}, 
    \isefxf{\xdp}{j}{\xp}, $\xap \doesnotefxenvyll{1} \xbp$,
    $\xap \lowerval{1} \xbp$, and $\minonep= \mintwoy$.  
    Therefore by \cref{ToGoodPart4},
    we can construct a new partition $\xt$ that is either \eftxf\ or \good\ 
    with $\phi(\xt)\geq \minonep = \mintwoy$.
\end{proof}

\begin{lemma}
\label{ToGoodPart from stage 2 case 1}
    Given a partition $\X$ in \stageTwo, suppose that $\Y$
    is a set of disjoint bundles such that $\ya = \xa \cup g$ with $g \in \xc \cup \xd$, $\yb=\xb$, 
    $\yc\sqcup \yd \subseteq (\xc \sqcup \xd) \setminus g$. Suppose further, that there exists some agent $i\neq 1$, such that
    \isefxffor{\yc}{i}{\Y}, and some agent $j\notin\{1,i\}$, such that \isefxffor{\yd}{j}{\Y}. 
    Then we can construct a new partition $\xt$ that is either \eftxf\ or \good\ 
    with $\phi(\xt)> \phi(\X)$.
\end{lemma}
\begin{proof}
    Since agent $1$'s valuation is non-degenerate, it holds that $v_1(\xb)>v_1(\xa)$, and $v_1(\xa \cup g)>v_1(\xa)$. Therefore, $\phi(\X)=v_1(\xa)<\mintwoy.$
    By \cref{ToGoodPart3}, we can construct a partition $\xt$ that is either \eftxf\ or \good\ with
    $\phi(\xt)\geq \mintwoy > \phi(\X)$.
\end{proof}

\begin{lemma}
\label{ToGoodPart5}
    Suppose $\X$ is a partition such that $\xc \sqcup \xd = \yc  \sqcup\yd$, 
    $(\xc,\xd) \doesnotefxenvyll{i} (\xa,\xb)$  for some agent $i \neq 1$, 
    and $(\yc,\yd) \doesnotefxenvyll{j}  (\xa,\xb)$ for some agent $j\notin \{1,i\}$.
    Then we can construct a new partition $\xt$ that is either \eftxf\ or \good\ 
    with $\phi(\xt)\geq \mintwo$.
\end{lemma}
\begin{proof}
    Run the PR algorithm w.r.t.\ agent $i$'s valuation on $(\xc,\xd)$ to get $(\ycp,\ydp)$. 
    W.l.o.g.\ assume that $\ydp\greatereqval{j} \ycp$.
    Let  $\xp = (\xa, \xb,\ycp,\ydp)$.
    Note that, since the valuation of agent $j$ is cancelable, it is MMS-feasible as well. Therefore, since $\ycp \sqcup \ydp = \xc \sqcup \xd = \yc \sqcup \yd$,
    we have that
    $\ydp \greatereqval{j}  \argmin_j(\yc,\yd)$. Since $\argmin_j(\yc,\yd)\doesnotefxenvyll{j} (\xa,\xb)$, it also holds  $\ydp\doesnotefxenvyll{j} (\xa,\xb)$, and 
    so \isefxffor{\ydp}{j}{\yp}. 
    Regarding agent $i$, since $\ycp,\ydp$ are given as the output of the PR algorithm, $\ycp \doesnotefxenvyll{i} \ydp$, and since the valuations are MMS-feasible, it holds that  
    $\ycp \greatereqval{i} \argmin_i(\xc,\xd)$. Since $\argmin_i(\xc,\xd) \doesnotefxenvyll{i} (\xa,\xb)$, it holds that $\ycp \doesnotefxenvyll{i} (\xa,\xb)$, and therefore, \isefxffor{\ycp}{i}{\yp}. Hence, by \cref{ToGoodPart3} 
    we can construct a new partition $\xt$ 
    that is either \eftxf\ or \good\ 
    with $\phi(\xt)\geq \mintwo$. 
\end{proof}

\subsection{Dealing with Partitions in \StageOne\ or \StageTwo}
\label{stage A or B section}

    In this subsection we prove (\cref{Stage 1 Proceed}) that if we start with a partition $\X$ in \StageOne, 
    then we either construct an \eftxf\ partition, or we construct another \goodpart\ 
    with higher potential.
    We further prove (\cref{ToPerfect1}) that if we start with a partition $\X$ in \StageTwo, 
    then we can either construct an \eftxf\ partition, or we construct another \goodpart\ 
    with higher potential, or we construct a partition in \stageTwoA with a potential not lower than before. Before proceeding with our first main theorem for this subsection, we restate the conditions of \stageOne.
    \\\\
    {\bf Conditions of \stageOne} (Restated).
    {\em 
        We say partition $\X$ is in \stageOne\ if $\xa= \argmin_1(X_1, X_2, X_3)$, 
        \areefxffor{\xa,\xb,\xc}{1}{\X},  and at least one of the following holds:
    \begin{itemize}
        \item i) $\xd$ is \bstf\ for at least one agent $j \neq 1$ and $\xd \setminus \hj \greatereqval{j} \xa\cup\hj$, 

        \item ii) $\xd$ is \bstf \ for at least two distinct agents $i,j \ne 1$.        
    \end{itemize}
        }

\begin{theorem}
    \label{Stage 1 Proceed}
    Given any partition $\X$ in \stageOne, 
    we can construct a new partition $\xt$ that is either \eftxf\ or \good\
    with $\phi(\xt)>\phi(\X)$.
\end{theorem}
\begin{proof}
    To prove this statement, we first address the case where $\X$ satisfies condition $i$ of $\stageOne$, and then condition $ii$.
    
    \textbf{\StageOnei} ($\xd$ is \bstf \ for some agent $j \neq 1$, $\xd  \setminus \hj  \greatereqval{j}  X_1 \cup \hj$): \\
    We define $$\xp = (\xa \cup \hj,\xb, \xc, \xd\setminus \hj)\,$$ 
    Then, since \isbstffor{\xd}{j}{\X},
    by \cref{best-feasible condition}, 
    agent $j$ does not \efxenvy\ the bundles $\xb,\xc$
    \with\ bundle  $\xd\setminus \hj=\xdp$.
    Also since
    $\xa \cup \hj \lowereqval{j} \xd\setminus \hj=\xdp$, 
    we get that \isefxffor{\xdp}{j}{\xp}. 
    Also we have that $\minthreep=\min (v_1(X_1 \cup \hj), v_1(\xbp), v_1(\xcp)) > v_1(X_1) = \phi(\X)$, since agent 1's valuation
    is non-degenerate. Therefore, 
    by \cref{ToGoodPart1}, 
    we can construct a new partition $\xt$ that is either \eftxf\ or \good\ 
    with $\phi(\xt)\geq \minthreep>\phi(\X)$, and the theorem follows for this case. 

    \textbf{\StageOneii} ($\xd$ is \bstf \ for at least two distinct agents $i,j \ne 1$):\\
    We consider the case that $\X$ is not in \stageOnei, otherwise we would show the theorem as in the previous case. Therefore, for any $u \in \{i,j\}$, we have
     $\xd  \setminus \hu \lowereqval{u} \xa \cup \hu$.
    By rearranging the bundles we define $$\xp  = (\xb, \xc, \xa, \xd)\,.$$ 
     Then, since \isbstffor{\xd}{u}{\X}, it holds that 
    $\xd \setminus \hu \doesnotefxenvyll{u} (\xb,\xc)$, and so it also holds that $\xa \cup \hu \doesnotefxenvyll{u} (\xb,\xc)$.
Note further that it holds that
     $\xa \sqcup \xd = (\xa \cup \hi) \sqcup (\xd  \setminus \hi) = (\xa \cup \hj) \sqcup (\xd \setminus \hj)$.
    Then, by \cref{ToGoodPart5}, 
    we can construct a new partition $\xt$ that is either \eftxf\ or \good\ 
    with $\phi(\xt) \geq \min(v_1(\xb),v_2(\xc)) > v_1(\xa) = \phi(\X)$, since agent 1's valuation
    is non-degenerate, and by this we complete the proof.
\end{proof}

Before proceeding with our second main theorem for this subsection, we restate the conditions of \stageTwo.\\\\
    {\bf Conditions of \stageTwo} (Restated).
    {\em 
        We say partition $\X$ is in \stageTwo\ if all the following hold:
        \begin{itemize}
        \item $\xa= \argmin_1(X_1, X_2)$, and \areefxffor{\xa,\xb}{1}{\X};
        \item \isefxf{\xc}{i}{\X} for some agent $i \ne 1$;
        \item $\xd \in \text{\efxfset}_j(\X)$ for some agent $j\notin \{1,i\}$;
        \end{itemize}
        }

\begin{theorem}
\label{ToPerfect1}
    Any partition $\X$ that is in \stageTwo\ is either \eftxf\ or in \stageTwoA.
\end{theorem}
\begin{proof}
    Since $\X$ is in \stageTwo, we have \areefxffor{(\xa,\xb)}{1}{\X}, 
    and for some distinct agents $i,j \ne 1$, we have
    \isefxffor{\xc}{i}{\X} and \isefxffor{\xd}{j}{\X}.
    Let the remaining agent be $u\notin\{1, i, j\}$. If any of the bundles $\xa, \xb$ is \efxf\ for agent $u$, 
    then $\X$ is \efxf\ (we could assign $\xc$ to $i$, $\xd$ to $j$, assign to $u$ their \efxf\ bundle from $\{\xa, \xb\}$, and the remaining bundle to agent 1) and therefore \eftxf. 
    If, on the other hand, neither $\xa$ nor $\xb$ are \efxf\ for agent $u$ in $\X$, 
    then the \bstf\ bundle for agent $u$ is either $\xc$ or $\xd$. Assume, w.l.o.g., that $\xd$ is \bstf\ for agent $u$ in $\X$. Since $\xa,\xb$ are not \efxf\ for $u$, by \cref{best props}, we get that $\xd \setminus \hu \greaterval{u} (X_1,X_2)$,
    so $\X$ is in \stageTwoA\ by setting $j=u$.
\end{proof}

\subsection{\NormalizationProcess}
\label{swap optimization section}
    
    Before going into detail regarding how we deal with partitions in the remaining stages, we dedicate this section to a subroutine, which we refer to as \normalizationprocess; this subroutine plays a crucial role in the rest of the paper. Given two bundles $S$ and $T$ ``assigned to'' two agents $i$ and $j$, this process swaps goods across the bundles based on the two agents' preferences, leading to Pareto improvements, until it reaches a local optimum w.r.t.\ such swaps.
    
\begin{definition}
    Given two disjoint bundles $S,T$ and two distinct agents $i,j$, we call the following operation \normalizationprocess\ of $(S,T)$ with respect to $(i,j)$: as long as there exists some $g \in S$ and $h \in T$ 
    such that $g\lowereqval{i} h$, $g \greatereqval{j} h$, and at least one of these two preferences is strict,
    swap these two \itms\ across the bundles.

    \begin{algorithm}[H]
    \SetAlgoRefName{Swap Optimize}
    \SetAlgoNoEnd
    \SetAlgoNoLine
    \DontPrintSemicolon
    \NoCaptionOfAlgo 
    \KwIn{$(S,T),(i,j)$} 
        \While{$\exists\ g \in S$, $h \in T$ 
                such that either $g\lowereqval{i} h$ and $g \greaterval{j} h$, or $g\lowerval{i} h$ and $g \greatereqval{j} h$,}{
            $S\gets (S\setminus g) \cup h$\\
            $T\gets (T\setminus h) \cup g$}
        \Return{$(S,T)$}        
    \caption{Swap Optimize}
    \label{shift sub-chain}
    \end{algorithm}
\end{definition}

\begin{remark}
    \NormalizationProcess\ runs polynomial time.
\end{remark}
    
    We now provide some lemmas regarding the \normalizationprocess\ process, which we then repeatedly use in our subsequent theorems and lemmas.
    
    Given a bundle $S\subseteq M$ and some natural number $k\leq |S|$, we use $G_k(S)= \{T\subseteq S : |T| = k\}$ to denote the set of subsets of $S$ of size $k$. Then, if $T^*\in \arg\max_{T \in G_k(S)}v_i (S\setminus T)$, we refer to $S\setminus T^*$, i.e., the best subset of $S$ that one can get after removing $k$ of its goods, as the \bestkremainer\ of $S$ w.r.t.\ agent $i$'s valuation. We refer to $w^k_i(S)=v_i (S\setminus T^*)$ as the {\em \bestkvalue} of a set $S$ for agent $i$; if $k>|S|$, we define the \bestkvalue\ of a set $S$ for any agent $i$ to be $0$.

\begin{lemma}
\label{best-k-value prop}
    Given a bundle $S$, two goods $g\in S$ and $h \not\in (S\setminus g)$, some agent $i$, and any $k\in \mathbb{N}$:
    \begin{itemize}
        \item if $g\lowereqval{i} h$, then $w^k_i(S) \leq w^k_i((S\setminus g)\cup h)$, and 
        \item if $g\greatereqval{i} h$, then $w^k_i(S) \geq w^k_i((S\setminus g)\cup h)$.
    \end{itemize}
\end{lemma}

\begin{proof}
    If $k=|S|$, we have that $w^k_i(S)=0$ and $w^k_i((S\setminus g)\cup h)=0$ so the statement obviously holds. Therefore, for the rest of this proof, we assume that $k<|S|$. We also suppose that $g\lowereqval{i} h$; the case when $g\greatereqval{i} h$ can be handled in the same manner.

    Let $S'$ be the \bestkremainer\ of $S$. 
    If $g\in S'$, by \cref{k-min with k-least}, all \itms\ of $S''=S\setminus S'$ are not more valuable than $g$ for agent $i$. Since $g\lowereqval{i} h$, the same holds for $h$ and, as a result, 
    the \bestkremainer\ of $(S\setminus g) \cup h$ is still 
    $((S\setminus g) \cup h)\setminus S'' = (S'\setminus g) \cup h)\greatereqval{i} S'$, 
    since $h\greatereqval{i} g$. 
    In the case that $g\notin S'$, meaning that $g\in S''$, 
    since $|(S''\setminus g)\cup h|=k$, we have that:
    \[w^k_i((S\setminus g)\cup h)\geq 
    v_i\Big(\big((S\setminus g)\cup h\big) \setminus 
    \big((S''\setminus g) \cup h\big)\Big) =v_i(S')=w_i^k(S).\qedhere\]
\end{proof}

\begin{lemma}
\label{Stability In Normalization}
    If we \normalize\ two bundles $(S,T)$ w.r.t.\ two agents $i,j$, 
    this weakly increases $w^k_i(S)$ and weakly decreases $w^k_i(T)$ 
    for every $k\in \mathbb{N}$. 
    Analogously, $w^k_j(S)$ weakly decreases and $w^k_j(T)$ weakly increases. 
\end{lemma}
\begin{proof}
    It suffices to show that the statement is true at the end of each swap of the \normalizationprocess\ process. 
    In every swap of this process, some item $g\in S$ is swapped with some item $h\in T$ such that
    $h\greatereqval{i}g$. So, by \cref{best-k-value prop}, we get that
    $w_i^k((S\setminus g)\cup h)\geq w_i^k(S)$ and $w_i^k((T\setminus h)\cup g)\leq w_i^k(T)$.
    A symmetric argument can be used for agent $j$.
\end{proof}

\begin{lemma}
\label{Normalization best-k-feasibility}
Suppose $\X$  is a partition and we \normalize\ two of its bundles, $(X_i, X_j)$, w.r.t.\ agents $(i,j)$ to get the bundles $(X'_i , X'_j)$, leading to partition $\X'$ with $X'_k = X_k$ for all $k\notin \{i,j\}$. Then, for every $u \in \{i,j\}$ and $k\in \mathbb{N}$, if  $X_u$ is \bstkf\ or \efxf\ for agent $u$ in $\X$, then $X^\prime_u $ is \bstkf\ or \efxf, respectively, for agent $u$ in $\xp$.  
\end{lemma}

\begin{proof}
    We prove the statement above for the case where $u=i$. By symmetry, the same arguments can be used for agent $j$, as well.
    
    If $X_i$ is \bstkf\
    in $\X$ for agent $i$, then it has 
    the maximum \bestkvalue\ among all the bundles in $\X$ for agent $i$. By \cref{Stability In Normalization}, the \bestkvalue\ of $i$ for $X_i'$ is weakly higher, and her value for all other bundles of $\X'$ is weakly lower than the corresponding bundles in $\X$. Therefore, we can conclude that $X'_i$ has the maximum \bestkvalue\ among the bundles in $\xp$, implying that $X_i'$ is the \bstkf\ bundle in $\xp$.
    
    If $X_i$ is \efxf\ for agent $i$ in $\X$, her value for $X_i$ is greater than her best-1-value for any other bundle in $\X$. By Lemma  \ref{Stability In Normalization}, $X_i'\greatereqval{i} X_i$ and the best-1-value of other bundles does not increase, therefore  $X_i'$ is \efxf\ for agent $i$ in $\xp$.
\end{proof}

\subsection{Dealing with Partitions in \StageTwoA}
\label{stage 2A section}

    In this subsection we prove (\cref{ToNormal}) that if we start with a partition $\X$ in \StageTwoA, 
    then we either construct an \eftxf\ partition, or we construct another \goodpart\ 
    with higher potential, or we construct a partition with a potential not lower than before
    in \stageOne\ or in \stageTwoNormal.
    We first restate the conditions of \stageTwoA. \\\\
    {\bf Conditions of \stageTwoA} (Restated).
    {\em 
        We say partition $\X$ is in \stageTwoA\ if all the following hold:
        \begin{itemize}
        \item $\xa= \argmin_1(X_1, X_2)$, and \areefxffor{\xa,\xb}{1}{\X};
        \item \isefxf{\xc}{i}{\X} for some agent $i \ne 1$;
        \item $\xd \in \text{\bstfset}_j(\X)$ for some agent $j\notin \{1,i\}$, and $\xd \setminus \hj \greaterval{j} (\xa, \xb)$;
        \end{itemize}
        }

\begin{theorem}
\label{ToNormal}
    Given any partition $\X$ in \stageTwoA, we can construct a new partition $\xt$ such that one of the following holds:
\begin{itemize}
    \item $\xt$ is \eftxf.
    \item $\xt$ is \good\ with $\phi(\xt)>\phi(\X)$.
    \item $\xt$ is in \stageTwoNormal \ with $\phi(\xt)\geq\phi(\X)$.
    \item $\xt$ is in \stageOne \ with $\phi(\xt)\geq\phi(\X)$.
\end{itemize}
\end{theorem}

We postpone the proof of \cref{ToNormal}, and we first prove the following useful lemma.

\begin{lemma}
\label{AgentAEnvies0}
    Suppose $\X$ is  a partition such that $\xa \lowerval{1} \xb$, 
    for some $j \ne 1$, agent $j$ does not \eftxenvy\ bundle $\xc$ \with\ bundle $\xd$, 
    and does not \efxenvy\ bundle $\xb$ \with\ bundle $\xd$, and $\xd \setminus \hj\greatereqval{j} \xa$. 
    Then if $\xa \efxenvyll{1} \xc$, we can construct a partition $\xt$ such that it is either \efxf\ or  \good\ 
    with $\phi(\xt)> \phi(\X)$
    or in \stageOne \ with $\phi(\xt) = \phi(\X)$.
\end{lemma}

\begin{proof}
    For $i \in \{2,3\}$, let $X'_i$ be a minimal subset of $X_i$ such that $X'_i \greaterval{1} \xa$ and
    let  $\unassignedItms = (\xc \setminus \xcp) \cup (\xb \setminus \xbp)$.
    The minimality of those sets gives that
    $(\xa,\xbp,\xcp) \doesnotefxenvyll{1} (\xa ,\xbp,\xcp)$.
    Note, that since $\xa \efxenvyll{1} \xc$, we have that $\unassignedItms \ne \emptyset$.
    We define $$\xp = (\xa, \xbp, \xcp, \xd \cup \unassignedItms)\,.$$  
    If any of the bundles $\xap, \xbp,\xcp$ is not \efxf\ for agent $1$ in $\xp$, 
    by \cref{ToGoodPart2},
    we get a partition $\xt$ such that it is either \efxf\ or  \goodpart \ such that
    $\phi(\xt) > \minfourp = v_1(\xa)=\phi(\X)$.
    So assume otherwise, i.e.,\ \isefxffor{(\xap, \xbp,\xcp)}{1}{\xp}. Since $\unassignedItms \ne \emptyset $ and $\xdp = \xd \cup \unassignedItms$, we get:
    \begin{center}
    $ \xdp \setminus (\hjp ,\Hjp) \greatereqval{j} (\xdp \setminus\unassignedItms) \setminus \hj = \xd  \setminus \hj \greatereqval{j} \xa \hspace{2mm} 
     \Rightarrow  \hspace{2mm}  \xdp \setminus \hjp \greatereqval{j} \xa \cup \hjp$
    \end{center}
    Also since $\unassignedItms \ne \emptyset$, by \cref{get item to bstf} and assumptions of the lemma,
    we get that \isbstffor{\xdp}{j}{\xp}.  Also we had that \areefxffor{\xap,\xbp, \xcp}{1}{\xp}. 
    Hence, partition $\xp$ is in \stageOne\ (where the first condition is satisfied), and we have $\phi(\xp) = v_1(\xa) =\phi(\X)$, which completes the proof.
\end{proof}

We are now ready to prove the main theorem of this subsection. 

\begin{proof}[\textbf{Proof of \cref{ToNormal}}]
    
    We will show that there exists either a partition satisfying the theorem's conditions or a partition $\Z$ in \stageTwoA, with $\phi(\Z)\geq \phi(\X)$ and $|Z_3|>|\xc|$. In the second case, we set $\Z$ as the initial partition and repeat; since the cardinality of the third bundle strictly increases at each iteration, 
    the procedure will terminate after at most $m$ steps to a partition satisfying the theorem's conditions. 
    
    Let $i,j$ be the agents as defined in \stageTwoA.
    After performing a \NormalizationProcess\ on $(\xc, \xd)$ w.r.t.\ agents $(i,j)$, let $(\xcp, \xdp)$ be the two new bundles. We define the partition 
    $$\xp = (\xa, \xb, \xcp, \xdp)\,.$$
    Note that in the \normalizationprocess, the cardinality of each bundle does not change, so $|\xcp| = |\xc|$.  
    By \cref{Stability In Normalization}, we get that: 
    \begin{center}
    $\xdp \setminus  \lst{j}{1}{\xdp}  \greatereqval{j}  \xd \setminus \lst{j}{1}{\xd} \greaterval{j} (\xa,\xb).$
    \end{center}
    Moreover, regarding $\X$ we had \isefxffor{\xc}{i}{\X}  and \isbstffor{\xd}{j}{\X}, so by \cref{Normalization best-k-feasibility},  
    we still have \isefxffor{\xcp}{i}{\xp} and \isbstffor{\xdp}{j}{\xp}. 
    
    If \areefxffor{\xap, \xbp}{1}{\xp}, 
    then $\xp$ would be in \stageTwoNormal \ with $\phi(\xp) = \phi(\X)$, and so the theorem would follow. If we assume otherwise, it should be either $\xap \efxenvyll{1} \xcp$, or $\xap \efxenvyll{1} \xdp$, since $\xap=\argmin_1(\xap,\xbp)$. If the first was true, the conditions of \cref{AgentAEnvies0} would be satisfied and the theorem would follow.  
    Thus, we assume that $\xap \efxenvyll{1} \xdp$.

    Let $\xdz$
    be a minimal subset of $\xdp$ such that $\xdz \greaterval{1} \xap$, and let $\unassignedItms = \xdp \setminus \xdz$. 
    Since $\xap \efxenvyll{1} \xdp$, we have that
    $\unassignedItms \ne \emptyset$. 
    We define $$\xz = (\xa, \xb, \xcp \cup \unassignedItms, \xdp \setminus \unassignedItms)\,. $$ 
    Since $\unassignedItms \ne \emptyset$, we get that $|\xcz|  >|\xcp| = |\xc|$. 
    Also since \isefxffor{\xcp}{i}{\xp}, by \cref{get item to bstf}, we get \isbstffor{\xcz}{i}{\xz}.
    If $\xaz \efxenvyll{1} \xcz$, then by \cref{ToGoodPart2}, we can construct a partition $\xt$ 
    that is either \eftxf\ or \good\ with $\phi(\xt) > \minfourz=v_1(\xa) = \phi(\X)$, and the theorem would follow.
    By assuming otherwise, since $\argmin_i(\xa, \xb, \xdz)=\xa$, it holds that \areefxffor{\xaz, \xbz, \xdz}{1}{\xz}.  
    
    If for any of the agents $j\notin \{1,i\}$, it was  \isbstf{\xcz}{j}{\xz}, 
    $\xz$ would be in the \stageOne\ (where the second condition is satisfied), with $\phi(\xz) = \phi(\X)$. If the other two agents, but $1$ and $i$, had different \bstf\ bundles in $\xz$, different from $\xcz$, then $\xz$ would be \efxf. So, consider the last case where both those agents have the same unique \bstf\ bundles in $\xz$, which is different from $\xcz$. Let this bundle be $Z_4$, and the two remaining ones among $\xaz, \xbz, \xdz$ be $Z_1,Z_2$ such that $Z_1\lowerval{1} Z_2$.
    We finally define the partition 
    $$\Z=(Z_1, Z_2, \xcz, Z_4)\,.$$
    For any agent $j\notin \{1,i\}$, \arenotefxffor{(Z_1,Z_2)}{j}{\Z}, and therefore, by \cref{best props}, $Z_4 \setminus \lst{j}{1}{Z_4} \greaterval{j}  (Z_1,Z_2).$ Moreover, since \isbstffor{\xcz}{i}{\xz}, it holds that \isbstffor{Z_3}{i}{\Z}, since $\Z$ is just a rearrangement of the bundles of $\xz$. Overall, $\Z$ is in \stageTwoA, with $\phi(\Z)\geq \minfourz=v_1(\xa) = \phi(\X)$, and $|Z_3|=|\xcz|>|\xc|$. By setting $\X=\Z$ and repeatedly transforming the partition as above, the procedure will terminate to a partition that satisfies the conditions of the theorem since, at each round, we strictly increase the cardinality of the third set. 
\end{proof}

\subsection{Dealing with Partitions in \StageTwoNormal}
\label{stage 2B section}
    
    In this subsection we prove (\cref{PerfectExists}) that if we start with a partition $\X$ in \stageTwoNormal, 
    then we either construct an \eftxf\ partition, or we construct another \goodpart\ 
    with higher potential, or we construct a partition with a potential not lower than before
    that is in \stageOne, or in \stageTwoC, or in \stageTwoPerfect. 
    We first restate the conditions of \stageTwoNormal. \\\\
{\bf Conditions of \stageTwoNormal} (Restated).
{\em 
    We say partition $\X$ is in \stageTwoNormal\ if all the following hold:
    \begin{itemize}
    \item $\xa= \argmin_1(X_1, X_2)$, and \areefxffor{\xa,\xb}{1}{\X};
    \item \isefxf{\xc}{i}{\X} for some agent $i \ne 1$;
    \item $\xd \in \text{\bstfset}_j(\X)$ for some agent $j\notin \{1,i\}$, and $\xd \setminus \hj \greaterval{j} (\xa, \xb)$;
    \item $(\xc, \xd)$ are \normalized\  w.r.t.\ agents $i,j$ (see Definition~\ref{normalized sets def});
    \end{itemize}
    }

\begin{theorem}
\label{PerfectExists}
Given any partition $\X$ in \stageTwoNormal, we can construct a new partition $\xt$ such that one of the following holds:
\begin{itemize}
    \item $\xt$ is \eftxf.
    \item $\xt$ is \good\ with $\phi(\xt)>\phi(\X)$.
    \item $\xt$ is in \stageOne \ with $\phi(\xt)\geq\phi(\X)$.
    \item $\xt$ is in \stageTwoC  \ with $\phi(\xt)\geq\phi(\X)$.
    \item $\xt$ is in \stageTwoPerfect  \ with $\phi(\xt)\geq\phi(\X)$.
\end{itemize}
\end{theorem}

We postpone the proof of \cref{PerfectExists} for later, and we first give three lemmas that are useful for proving \cref{PerfectExists}.

\begin{lemma}
    \label{ToNormal0}
    Suppose $\X$ is a partition in \stageTwoA \ with corresponding agents $i,j$
    such that $\xc \setminus \gi \greaterval{i} (\xa,\xb)$. 
    Then we can construct a new partition $\xt$ such that one of the following holds:

    \begin{enumerate}
        \item $\xt$ is \eftxf. 
        \item $\xt$ is \good\ with $\phi(\xt)>\phi(\X)$.
        \item $\xt$ is in \stageOne \ with $\phi(\xt)\geq \phi(\X)$.
        \item $\xt$ is in \stageTwoPerfect \ with $\phi(\xt)\geq \phi(\X)$. 
    \end{enumerate}
\end{lemma}

\begin{proof}
    We remind the reader that a partition is in \stageTwoPerfect\ if it is in \stageTwoNormal\ and also it holds that
    $\xc \setminus \gi \greaterval{i} (\xa,\xb)$ where $i$ is the agent as defined in \stageTwoNormal.
    We now proceed with the proof. We \normalize\  $(\xc,\xd)$ w.r.t.\ agents $(i,j)$ to get $(\xcp,\xdp)$ and we keep the first two bundles the same, i.e., $\xap=\xa$ and $\xbp =\xb$. In the resulting partition $\xp = (\xap, \xbp, \xcp, \xdp)$, we have:
    \begin{center}
        $\xdp \setminus \hjp \greatereqval{j}  \xd \setminus \hj \greaterval{j}  (\xa,\xb)$    \hspace{8mm} 
            (By \cref{Stability In Normalization})

        $\xcp \setminus \gip \greatereqval{i} \xc \setminus  \gi  \greaterval{i} (\xa,\xb)$  \hspace{8mm} 
            (By \cref{Stability In Normalization})
\end{center}
    By \cref{Normalization best-k-feasibility}, 
    \isefxf{\xcp}{i}{\xp} and \isbstffor{\xdp}{j}{\xp} and so it is \isefxf{\xdp}{j}{\xp}, as well.
    If \areefxffor{\xa,\xb}{1}{\xp},  then $\xp$ is in \stageTwoPerfect \ and $\phi(\xp)=v_1(\xa)$,
    so by setting $\xt$, the lemma would follow. If this is not the case, then it is either $\xa \efxenvyll{1} \xcp$ or $\xa \efxenvyll{1} \xdp$. 
    In each case, using \cref{AgentAEnvies0}, 
    we can construct a partition $\xt$ that it is either \efxf\ or  \good\ 
    with $\phi(\xt)> v_1(\xa)$ or in \stageOne \ with $\phi(\xt) = v_1(\xa)$, and the lemma follows again.
\end{proof}

\begin{lemma}
\label{SeperateBest-2-feasible}
    Suppose $\X$ is a partition such that $\xa \lowerval{1} \xb$ and 
    for two distinct agents $i,j \ne 1$ the following two conditions hold:
    \begin{align}
        (\xa,~\xb,~\xd \setminus \{\hi,\Hi\}) \lowereqval{i} \xc \setminus \{\gi,\Gi\} \label{sep:1} \\ 
        (\xa,~\xb,~\xc \setminus \{\gj,\Gj\}) \lowereqval{j} \xd \setminus \{\hj,\Hj\} \label{sep:2}
    \end{align}
    Then we can construct a new partition $\xt$ that is either \eftxf \ or \good\
    with $\phi(\xt)>\phi(\X)$.
\end{lemma}
\begin{proof}
    We \normalize\ $(\xc, \xd) $ w.r.t.\ agents $(i,j)$ to get $(\xcp, \xdp)$. 
    By \cref{Stability In Normalization}, all upper conditions hold for $\xp$ too.
    Therefore, for simplicity, suppose that originally $\xc,\xd$ are \normalized\  w.r.t.\ agents $(i,j)$.
    We define the following set of distinct bundles: 
    \begin{center}
        $\Y = \Big(\xa \cup \bigerlstitem{i}{1}{\{\gi,\hj\}},\hspace{2mm} \xb,
        \hspace{2mm} \xc \setminus \gi,\hspace{2mm} \xd \setminus \hj \Big)$
    \end{center}
    with the set $\unassignedItms = \{ \bigerhstitem{i}{1}{\{\gi,\hj\}} \}$ of unused \itms. 
    We now show that \isefxffor{\yc}{i}{\Y} and \isefxffor{\yd}{j}{\Y}.
    Using \cref{sep:1}, \cref{sep:2}, and \cref{cancelable prop 1} we get:
    \begin{align}
        (\xa\cup\gi,~ \yb, ~\xd \setminus \{\hi,\Hi\}) \lowereqval{i} \xc \setminus \gi =\yc  \label{sep:3} \\
        (\xa\cup\hj,~ \yb, ~\xc \setminus \{\gj,\Gj\}) \lowereqval{j} \xd \setminus \hj =\yd  \label{sep:4}
    \end{align}
    By using the fact that $\bigerlstitem{i}{1}{\{\gi,\hj\}} \lowereqval{i} \gi$, and \cref{sep:3} we get:
    \begin{align}
        &\ya =\xa \cup \bigerlstitem{i}{1}{\{\gi,\hj\}} \lowereqval{i} \xa \cup \gi
        \lowereqval{i} \yc\\
        &\yd\setminus \hiy \lowereqval{i} \xd \setminus \{\hi,\Hi\} \lowereqval{i} \yc
    \end{align}
    Therefore, \isefxffor{\yc}{i}{\Y}.

    Now we prove \isefxffor{\yd}{j}{\Y}.
    With similar arguments we have  $\yd \doesnotefxenvyll{j} (\yb,\yc)$.
    So it remains to show that $\yd \doesnotefxenvyll{i} \ya$.
    By \cref{sep:4}, we have $\xa \cup \hj \lowereqval{j} \yd$. 
    Now we prove that $\ya = \xa \cup \bigerlstitem{i}{1}{\{\gi,\hj\}} \lowereqval{j} \yd$.
    If $\bigerlstitem{i}{1}{\{\gi,\hj\}} = \hj$, the argument holds obviously.
    If $\bigerlstitem{i}{1}{\{\gi,\hj\}} = \gi$, 
    we get $\gi \lowereqval{i} \hj$ therefore by \normalizationprocess\ of $(\xc,\xd)$ w.r.t.\ $(i,j)$,
    we get $\gi \lowereqval{j} \hj$, so  
    we get $\xa \cup \gi \lowereqval{j} \xa \cup \hj \lowereqval{j} \yd$.
    So \isefxffor{\yd}{j}{\Y}.
    
    Therefore by \cref{ToGoodPart from stage 2 case 1}, we can construct a new partition $\xt$ of all the items 
    that is either \eftxf \ or \good\ 
    with $\phi(\xt)\geq \mintwoy > \phi(\X)$.
\end{proof}

\begin{lemma}
\label{bipartite  graph}
    Suppose that $\X$ is a partition such that
    \areeftxffor{\xa,\xb}{1}{\X}, for distinct agents $i,j\ne 1$,
    \iseftxffor{\xc}{i}{\X}, \iseftxffor{\xd}{j}{\X},
    agents $i,j$ have at least two \eftxf\ bundles in $\X$
    and at least one of the bundles $(\xa,\xb)$ is \eftxf\ for either agent $i$ or $j$.
    Then $\X$ is \eftxf.
\end{lemma}
\begin{proof}
    By lemma's statement, one of the bundles \xab is \eftxf\ for either agent $i$ or $j$. Assume, w.l.o.g. that this is agent $i$ (the conditions of \stageTwo\ for agents $i,j$ are symmetric).
    Therefore, suppose that $\argmax_i (\xa, \xb)$ is \eftxf\ for agent $i$. 
    
    We construct a bipartite graph $G(A,X,E)$, where $A$ is the set of agents, $X$ is the set of bundles of $\X$, and $(i,X_j) \in E$ whenever agent $i$ is \eftxf\ with $X_j$. It is 
    easy to verify that the condition of Hall's Theorem is satisfied and there exists a perfect matching corresponding to an EF2X allocation. To see that the condition holds, notice 
    that agent 1 is \eftxf\ with $\xa$ and $\xb$, agent $i$ is \eftxf\ with one of the bundles of \xab and $\xc$, and agent $j$ is \eftxf\ with $\xd$ and some other bundle.
    So, for any set of agents of cardinality 1, 2, or 4, the condition is trivially satisfied. For any set of three agents, if agent 1 is included, then $\xa$ and $\xb$ belong to their neighbors, and there should be one of the $\xc$ or $\xd$ as the neighbor of the other two agents. If agent 1 is not included in the set, then the sets $\argmax_i (\xa, \xb)$, $\xc$ and $\xd$ belong to their neighbors. So, overall, the condition of Hall's Theorem is satisfied.   
\end{proof}

    Now we are ready to prove \cref{PerfectExists}.

\begin{proof}[\textbf{Proof of \cref{PerfectExists}}]
    We define an integer $ p \in [0, |\xd|-1]$ to be the greatest integer such that for  
    $\unassignedItms= \{\lst{j}{1}{\xd}, \lst{j}{2}{\xd},\ldots, \lst{j}{p}{\xd}\}$, 
    $\xdp = \xd \setminus \unassignedItms$, and $\xcp = \xc \cup \unassignedItms$,
    the following hold:
    \begin{align}
    \xdp \setminus \hjp  &\greaterval{j} (\xa, \xb)  \mbox{ and }
    \xdp  \doesnotefxenvyll{j} \xcp \label{eq:minimalityForX4}.
    \end{align}
    It may be that $S=\emptyset$, in which case we consider $p=0$. 
    Note that such $p$ exists since for $p = 0$, the above two conditions hold, and for $p = |\xd|-1$, the first condition doesn't hold.

    Since $p$ is maximal, for any integer greater than $p$, at least one of the conditions in \cref{eq:minimalityForX4} does not hold, so it is:
    \begin{align}
    \xdp \setminus \{\hjp, \Hjp\} &\lowereqval{j} \argmax_j (\xa, \xb) \mbox{\qquad or \qquad}
    \xdp \setminus \hjp  \efxenvyll{j} \xcp \cup \hjp\,. \label{eq:conditionsWRTp} 
    \end{align}

    Let $\xp = (\xa, \xb, \xc \cup \unassignedItms, \xd \setminus \unassignedItms)$.  
    Then by upper inequalities, we get \isefxffor{\xdp}{j}{\xp}, 
    and since for $q \in \{1,2,4\}$, we have $X'_q \subseteq X_q$, $\xc \subseteq \xcp$, 
    and since \isefxffor{\xc}{i}{\X}, 
    we get that \isefxffor{\xcp}{i}{\xp}.
    
    If any of the bundles $\xap,\xbp$ is not \efxf\ for agent $1$ in ${\xp}$, 
    then since $(\xa,\xb) \doesnotefxenvyll{1} \xd$ and 
    $\xdp \subseteq \xd$, 
   it holds that $(\xap,\xbp) \doesnotefxenvyll{1} \xdp$, 
    and since $\xa\lowerval{1}\xb$ it should be that $\xap \efxenvyll{1} \xcp$. Then, by \cref{AgentAEnvies0}, the theorem would follow. So, suppose otherwise, i.e., \isefxffor{(\xap,\xbp)}{1}{\xp}.

    If $\xp$ is \efxf, the theorem would follow, since $\phi(\xp)=v_1(\xa)=\phi(\X)$. So, suppose  $\xp$ is not \efxf. 
    Let $u$ be the remaining agent $u \not\in \{1,i,j\}$.
    Note that $(\xap, \xbp)$ are not \efxf\ for agent $u$, and therefore not \bstf \ in $\xp$, otherwise $\xp$ would be \efxf. We will consider the other two cases, \isbstf{\xcp}{u}{\xp} and \isbstf{\xdp}{u}{\xp}, separately, and we will further use the fact that since \isnotefxf{(\xap, \xbp)}{u}{\xp}: 
    \begin{align}
        \xcp \setminus \gup \greaterval{u} (\xap,\xbp) &\mbox{\quad if \isbstf{\xcp}{u}{\xp}} \, \text{ and}\label{eq:uIsNotEFXFwithX1X2-forX3}\\
        \xdp \setminus \hup \greaterval{u} (\xap,\xbp)  &\mbox{\quad if \isbstf{\xdp}{u}{\xp}} \,. \label{eq:uIsNotEFXFwithX1X2-forX4}
    \end{align} 
        
    If \isbstf{\xcp}{u}{\xp}, it further hold that
     $\xcp \setminus \gup \greaterval{u} (\xap,\xbp)$ (by \eqref{eq:uIsNotEFXFwithX1X2-forX3}),
    $\xdp \setminus \hjp \greaterval{j} (\xap,\xbp)$ (by \eqref{eq:minimalityForX4}), and $\phi(\X)=\phi(\xp)$. 
    By \cref{ToNormal0}, (after swapping $\xc$ and $\xd$)
    we can construct a new partition $\xt$ that has one of the desired properties of the theorem, and the proof would be complete.
    
    If \isbstffor{\xdp}{u}{\xp}, by \cref{best props}, it also holds that \isefxffor{\xdp}{u}{\xp}.  
    Note that $(\xa,\xb)$ are not \efxf\ for agent $j$ in $\xp$, otherwise $\xp$ would be \efxf; therefore,  
    $(\xa,\xb)$ are not \bstf\ for agent $j$ in $\xp$. 
    If \isbstffor{\xcp}{j}{\xp}, 
    then it would be that $\xcp \setminus \gjp \greaterval{j} (\xa, \xb)$, 
    and since $\xdp \setminus \hup \greaterval{u} (\xa, \xb)$ (by \eqref{eq:uIsNotEFXFwithX1X2-forX4}), the theorem would follow by \cref{ToNormal0}.
    So assume that \isbstffor{\xdp}{j}{\xp}, and for the same reason it holds that $\xdp \setminus \hjp \greaterval{j} (\xa, \xb)$.

    Overall, that far, for partition $\xp = (\xap, \xbp, \xcp, \xdp)$ we assume that
    $\xap \lowerval{1} \xbp$, \areefxffor{(\xap,\xbp)}{1}{\xp}, 
    \isefxffor{\xcp}{i}{\xp}, \isbstffor{\xdp}{j}{\xp}, and $\xdp \setminus \hjp \greaterval{j}  (\xap, \xbp)$, otherwise the theorem follows. It further holds that $(\xcp, \xdp)$ are \normalized\  with respect to $(i,j)$. The reason is that by the definition of $\unassignedItms$, 
    we have that $\forall x \in \unassignedItms, \forall y \in \xdp : x \lowereqval{j} hy$, 
    and since $(\xc,\xd)$ was \normalized\  w.r.t.\ agents $(i,j)$, 
    we have that $(\xcp, \xdp)$ are \normalized\  w.r.t.\ agents $(i,j)$. Therefore, $\xp$ is in \stageTwoNormal\ with $\phi(\xp)=\phi(\X)$ (otherwise the theorem follows).

    If $\xcp \setminus \gip \greaterval{i} (\xap, \xbp)$, then $\xp$ is in \stageTwoPerfect, and the theorem would follow. So, we suppose that 
    \begin{equation}
        \xcp \setminus \gip \lowereqval{i} \argmax_i(\xap, \xbp)\,. \label{eq:iPrefersX_1X_2}
    \end{equation}

    Recall now that at least one of the conditions in \cref{eq:conditionsWRTp} holds.
    If the first condition of \eqref{eq:conditionsWRTp} is true, then it is $\xdp \setminus \{\hjp, \Hjp\} \lowereqval{j} \argmax_j (\xa, \xb)$. If the second condition of \eqref{eq:conditionsWRTp} is true, then it is $\xdp \setminus \{\hjp, \Hjp\} \lowerval{j} \xcp$,
    since otherwise we would get that $\xdp \setminus \hjp \greatereqval{j} \xcp \cup \hjp$,
    which is a contradiction to that condition. Therefore, overall it should be that   \isefxffor{\argmax_j (\xap, \xbp, \xcp)}{j}{\xp}, and since $\xp$ is in \stageTwoNormal, it also holds that \isefxffor{\xdp}{j}{\xp}; then it is also \areeftxffor{\argmax_j (\xap, \xbp, \xcp), \xdp}{j}{\xp}.
    
    In order to complete the proof, we will divide it into cases and sub-cases regarding $\xap, \xbp$ being \eftxf\ for agents $i,j$ or not.
    
    \noindent\textbf{$\bullet$ \caseA. \iseftxffor{\mathbf{\argmax_i(\xap, \xbp)}}{i}{\xp}:}
    The conditions of \cref{bipartite graph} are satisfied for $\xp$, so $\xp$
    \eftxf\ and the proof completes.
     
    \noindent\textbf{$\bullet$ \caseB. \arenoteftxffor{\mathbf{\xap, \xbp}}{i}{\xp}:}
     If $\xdp \setminus \{\hip, \Hip\} \lowereqval{i} (\xap, \xbp)$, by \cref{eq:iPrefersX_1X_2} it would be that \iseftxffor{\mathbf{\argmax_i(\xap, \xbp)}}{i}{\xp}, which contradicts the condition of this case. Therefore, it is
     $\xdp \setminus \{\hip, \Hip\} \greaterval{i} (\xap, \xbp)$, and by \cref{eq:iPrefersX_1X_2}, we can infer that \isbsttwoffor{\xdp}{i}{\xp}, and therefore, \areeftxffor{\xdp}{i}{\xp}. Since $\xp$ is in \stageTwoNormal, it also holds that \areeftxffor{\xcp}{i}{\xp}
    
    \textbf{$\bullet$ \subcaseA. \iseftxffor{\mathbf{\argmax_j(\xap,\xbp)}}{j}{\xp}:}
    The conditions of \cref{bipartite  graph} are satisfied for $\xp$, so $\xp$ is \eftxf\ and the proof completes.
    
    \textbf{$\bullet$ \subcaseB. \arenoteftxffor{\mathbf{\xap,\xbp}}{j}{\xp}:}
    Suppose that \isbsttwoffor{\xcp}{j}{\xp}, and 
    we have further showed for \caseB\ that \isbsttwoffor{\xdp}{i}{\xp}. 
    Since both $(\xap,\xbp)$ are not \eftxf\ for agents
    $i,j$, we get that $(\xap,\xbp) \lowerval{j} \xcp \setminus \{\gjp, \Gjp\}$ and $(\xap,\xbp) \lowerval{i} \xdp \setminus \{\hip, \Hip\}$. Overall, the conditions of \cref{SeperateBest-2-feasible} are satisfied, and  
    we can construct a new partition $\xt$ that is either \eftxf\ or \good\
    with $\phi(\xt)>\phi(\xp)=\phi(\X)$, and the proof completes.
    Consider now the other case that \arenotbsttwoffor{\xcp}{j}{\xp}, and therefore it is \isbsttwoffor{\xdp}{j}{\xp}.
    Since $(\xap,\xbp)$ are not \eftxf\ for agent
    $j$, it is $(\xap,\xbp) \lowerval{j} \xdp \setminus \{\hjp, \Hjp\}$, which means that the first condition of \cref{eq:conditionsWRTp} does not hold, so the second one should hold, i.e., 
    $\xdp \setminus \hjp  \efxenvyll{j} \xcp \cup \hjp$.
    Hence, partition $\xp$ is in \stageTwoC\ (by conditions showed in \subcaseB\ and \cref{eq:iPrefersX_1X_2}), and since $\phi(\xp)=\phi(\X)$, the theorem follows.
\end{proof}

\subsection{Dealing with Partitions in \StageTwoC}
\label{stage 2c section}

    In this subsection we prove (\cref{Perfect}) that if we start with a partition $\X$ in \stageTwoC, 
    then we either construct an \eftxf\ partition, or 
    we construct another \goodpart\ with higher potential, 
    or we construct a partition in \stageOne\ with the same potential.
    We first restate the conditions of \stageTwoC. \\\\
    {\bf Conditions of \stageTwoC} (Restated).
    {\em 
        We say partition $\X$ is in \stageTwoPerfect\ if all the following hold:
        \begin{itemize}
        \item $\xa= \argmin_1(X_1, X_2)$, and \areefxffor{\xa,\xb}{1}{\X};
        \item \isefxf{\xc}{i}{\X} for some agent $i \ne 1$, 
        and $\xc \setminus \gi \lowereqval{i} \argmax_i(\xa, \xb) \lowerval{i} \xd \setminus \{\hi, \Hi\}$;
        \item $\xd \in \text{\bstfset}_j(\X)$ for some agent $j\notin \{1,i\}$, 
        and $(\xa, \xb) \lowerval{j} \xd \setminus \{\hj, \Hj\}$;
        \item $\xd \setminus \hj \efxenvyll{j} \xc \cup \hj$; 
        \item $(\xc, \xd)$ are \normalized\  w.r.t.\ agents $i,j$ (see Definition~\ref{normalized sets def});
        \end{itemize}
        }

\begin{theorem}
\label{Perfect}
    Given any partition $\X$ in \stageTwoC, 
    we can construct a new partition $\xt$ such that one of the following holds:
    \begin{itemize}
        \item $\xt$ is \eftxf.
        \item $\xt$ is \good\ with $\phi(\xt)>\phi(\X)$.
        \item $\xt$ is in \stageOne \ with $\phi(\xt)=\phi(\X)$. 
    \end{itemize}
\end{theorem}
\begin{proof}
    We divide the proof into three cases:
    \begin{enumerate}
         \item $\Hi \greaterval{i} \gi$    
         \item $\Hi \lowereqval{i} \gi \hspace{0.2 cm}$ and $\hspace{0.2 cm} \hj \greatereqval{j} \gj$
         \item $\hj \lowerval{j} \gj$.
    \end{enumerate}
    Then we handle these three cases in \cref{DerivingPerfect1}, \cref{DerivingPerfect2}, and \cref{DerivingPerfect3}, respectively.
\end{proof}

\begin{lemma}
\label{DerivingPerfect1}
    Suppose $\X$ is a partition in \stageTwoNormal \ with corresponding agents $i,j$ such that:
    \begin{align}
         (\xa, \xb) &\lowerval{i} \xd \setminus \{\hi, \Hi\}\label{p1:xa xb <i xd--}  \\ 
         (\xa,\xb)  &\lowerval{j} \xd \setminus \{\hj, \Hj\}\label{p1:xa xb <j xd--}   \\
          \gi &\lowerval{i} \Hi \label{p1:gi <i Hi} 
    \end{align}
    \donef
\end{lemma}

\begin{proof}
    By using \cref{p1:xa xb <i xd--}, \cref{p1:xa xb <j xd--} and \cref{cancelable prop 1} 
    we get the following equations:
    \begin{align}
        \xa \cup \hi \lowereqval{i} \xa \cup \Hi &\lowereqval{i} \xd \setminus \hi \label{p:xa xb <i xd--} \\
        \xa \cup \hj \lowereqval{j} \xa \cup \Hj &\lowereqval{j} \xd \setminus \hj \label{p:xa xb <j xd--}
    \end{align}
    \noindent\textbf{$\bullet$ \caseA. $\mathbf{\hj = \hi}$:}
    We define 
    \begin{center}
        $\xp = (\xa \cup \hj, ~~ \xb, ~~ \xc,~~ \xd \setminus \hj)$, 
    \end{center}
    then we show that \isefxffor{\xdp}{j}{\xp}. 
    The fact that \isbstffor{\xd}{j}{\X} along with \cref{best-feasible condition} results in
    $\xdp= \xd \setminus \hj \doesnotefxenvyll{j} (\xb,\xc)$. 
    Also, by \eqref{p:xa xb <j xd--}, we get that $\xdp \greatereqval{j} \xap$.
    Hence, \isefxffor{\xdp}{j}{\xp}.
    Next we show that \isefxffor{\xcp}{i}{\xp}.
    By $\hj = \hi$, \cref{p:xa xb <i xd--}, and the fact that \isefxffor{\xc}{i}{\X}, we get:
    \begin{equation}
        \xap=\xa \cup \hj = \xa \cup \hi \lowereqval{i}  \xd \setminus \hi \lowereqval{i} \xc=\xcp \notag  
    \end{equation}
    Also by  \isefxffor{\xc}{i}{\X}, we get that $\xcp=\xc \doesnotefxenvyll{i} (\xbp,\xdp)$.
    Hence, \isefxf{\xcp}{i}{\xp}.
    \donew

    \vspace{5 pt}
    \noindent\textbf{$\bullet$ \caseB. $\mathbf{\hj \ne \hi}$:}
    In this case, $\hj$ is either the second less valuable \itm\ for agent $i$ or more valuable than that, i.e., $\Hi \lowereqval{i} \hj$. By \eqref{p1:gi <i Hi} we get:
    \begin{equation}\label{p1:gi <i hj}
        \gi \lowerval{i} \Hi \lowereqval{i} \hj.        
    \end{equation}
    In this case, we define:
    \begin{center}
    $\xp=\big(\xa \cup \gi ,\ \xb,\ (\xc \setminus \gi) \cup \hj,\ \xd \setminus \hj \big)$.    
    \end{center}
    We first show that \isefxffor{\xcp}{i}{\xp}. 
    By using \isefxffor{\xc}{i}{\X} and \cref{p1:gi <i hj}, we get:
    \begin{equation} \label{p: xcp >i xd-}
        \xd \setminus  \hj \lowereqval{i} \xc = (\xc \setminus \gi) \cup \gi \lowereqval{i} 
        (\xc \setminus \gi) \cup \hj = \xcp        
    \end{equation}

    By Equations \eqref{p1:gi <i Hi}, \eqref{p:xa xb <i xd--}, \isefxffor{\xc}{i}{\X} and \cref{p: xcp >i xd-} we get:
    \begin{equation}  \label{p: xcp >i xap}
        \xap = \xa \cup \gi \lowereqval{i} \xa \cup \Hi \lowereqval{i} 
        \xd \setminus \hi \lowereqval{i} \xc \lowereqval{i} \xcp
    \end{equation}

    Finally, by \isefxffor{\xc}{i}{\X} and the fact that $\xc \lowereqval{i} \xcp$ (by \eqref{p: xcp >i xd-}) we get that $\xb  \doesnotefxenvygg{i} \xcp$, which in turns gives along with Equations \eqref{p: xcp >i xd-} and  \eqref{p: xcp >i xap}, that \isefxffor{\xcp}{i}{\xp}. 
    
    We now turn our attention to agent $j$. Since $\X$ is in \stageTwoNormal,
    $(\xc,\xd)$ are \normalized\ w.r.t.\ agents $(i,j)$,
    and since $\gi \lowerval{i} \hj$ (by \eqref{p1:gi <i hj}), it holds that $\gi \lowereqval{j} \hj$. By using also  
     \cref{p:xa xb <j xd--}, we get that: 
    \begin{equation}\label{P: xdp 1g}
        \xap = \xa \cup \gi \lowereqval{j} \xa \cup \hj \lowereqval{j} \xd \setminus \hj =\xdp
    \end{equation}
    Additionally, by \cref{p1:xa xb <j xd--} we get:
    \begin{equation}\label{P: xdp 2g}
        \xbp=\xb \lowerval{j} \xdp.
    \end{equation}
    If it was also the case that $\xdp \greatereqval{j} \xcp$, it would be that \isefxffor{\xdp}{j}{\xp}, and the lemma would follow by \cref{ToGoodPart from stage 2 case 1}. So, assume otherwise, i.e.,
    \begin{equation}\label{eq:x4<jX3}
    \xdp \lowereqval{j} \xcp.
    \end{equation}
    This in turn gives:
    \begin{equation}\label{p1:not sat}
    \xd \setminus \hj\lowereqval{j} (\xc \setminus \gi)\cup \hj \lowereqval{j}  
    (\xc \setminus \gi) \cup \hi.
    \end{equation}
    We further define the following allocation:
    \begin{center}
        $\xz =(\xa \cup \gi,~~ \xb,~~ (\xc \setminus \gi) \cup \hi,~~ \xd \setminus \hi)$.
    \end{center}
    We first show that \isefxffor{\xcz}{j}{\xz}. By \cref{p1:not sat}, we have:
    \begin{equation*}\label{p:xdz<j xcz}
     \xdz= \xd \setminus \hi \lowereqval{j} \xd \setminus \hj\lowereqval{j} \xcz\,,
    \end{equation*}
    which in combination with \cref{p1:xa xb <i xd--} gives that:
    \begin{equation*}\label{p:xaz xbz <j xcz}
    (\xa \cup \gi,\xb)  \lowereqval{j} \xd \setminus \hj \lowereqval{j} \xcz.
    \end{equation*}
    Overall, it holds that
    \isefxffor{\xcz}{j}{\xz}.
    
    Considering now agent $i$, by \cref{p1:xa xb <i xd--} it holds that
    \begin{equation}\label{p:xaz-}
        \xbz=\xb \lowereqval{i} \xd \setminus \{\hi,\Hi\}  \lowereqval{i} \xd \setminus \hi = \xdz\,,        
    \end{equation}
    and by \cref{p1:gi <i Hi} and \cref{p:xa xb <i xd--}, it holds that
    \begin{equation}\label{p:xabz-}
    \xaz=\xa \cup \gi \lowereqval{i} \xa \cup \Hi \lowereqval{i} \xd \setminus \hi= \xdz\,.  
    \end{equation}
    If it was also the case that $\xdz \greatereqval{i} \xcz$, it would be that \isefxffor{\xdz}{i}{\xz}, and the lemma would follow by \cref{ToGoodPart from stage 2 case 1}. So assume otherwise, i.e., $\xdz \lowereqval{i} \xcz$.   
    By using also Equations 
    \eqref{p:xaz-} and \eqref{p:xabz-}, it holds that
    \begin{align}
    (\xap,\xbp) = (\xaz,\xbz) &\lowereqval{i} \xdz \lowereqval{i} \xcz \,.
    \end{align}
    Moreover, by Equations \cref{P: xdp 1g}, \cref{P: xdp 2g} and  \cref{p1:not sat} \eqref{eq:x4<jX3} gives:
    \begin{align}
    (\xap,\xbp) = (\xaz,\xbz) \lowereqval{j} \xdp \lowereqval{j} \xcp\,.
    \end{align}
    Since it is also true that $\xcp \sqcup \xdp = \xcz \sqcup \xdz$, 
    by \cref{ToGoodPart5}, 
    we can construct a new partition $\xt$ that is either \eftxf\ or \good\ 
    with $\phi(\xt)\geq \mintwop > \phi(\X)$; the last inequality is due to \cref{non degenerate} and due to the fact that $\xap \greatereqval{1} \xa$ and $\xb \greatereqval{1} \xa$ (by \stageTwoNormal). 
    \end{proof}
    
\begin{lemma}
\label{DerivingPerfect2}
    Suppose $\X$ is a partition in \stageTwoC\ such that  
    $\Hi \lowereqval{i} \gi$ and $\hj \greatereqval{j} \gj$.
    \donef
\end{lemma}
\begin{proof}
    Let $\xp$ be the partition with the following  bundles:
    \begin{align*}
    \xap &= \xa \cup \lstitem{j}{1}{\{\hi, \Hi\}} \\
    \xbp &= \xb \\
    \xcp &= (\xc \setminus \gj)\cup \hstitem{j}{1}{\{\hi, \Hi\}} \\
    \xdp &= (\xd \setminus\{\hi, \Hi\}) \cup \gj
    \end{align*}
        
    \noindent We will show that \isefxffor{\xcp}{j}{\xp} and \isefxffor{\xdp}{i}{\xp}.
    Since $\X$ is in \stageTwoC, the following holds:
    \begin{equation}\label{p2:xd-->j xc}
        \xc \setminus \gj \lowereqval{i} \xc \setminus \gi \lowereqval{i} \argmax_i (\xa,\xb)
        \lowereqval{i} \xd \setminus\{\hi, \Hi\}.
    \end{equation}
    Due to the lemma's condition,  $\Hi \lowereqval{i} \gi$, we get that:
    \begin{equation}\label{p2: hH}
        (\hi, \Hi) \lowereqval{i} \Hi \lowereqval{i} \gi \lowereqval{i} \gj.
    \end{equation}
    By combining \cref{p2:xd-->j xc}, \cref{p2: hH}, it holds that $(\xap, \xbp, \xcp) \lowereqval{i} \xdp$ 
    and so, \isefxffor{\xdp}{i}{\xp}. 
    
    Next, we show that \isefxffor{\xcp}{j}{\xp}.   
    Since $\X$ is in \stageTwoC, we have that
    $\xd \setminus \hj \efxenvyll{j} (\xc \cup \hj)$. By the lemma's condition, $\gj \lowereqval{j} \hj$, 
    it holds that the least valued item in $\xc \cup \hj$ is $\gj$, and therefore, $\xd \setminus \hj  \lowerval{j} (\xc \setminus \gj) \cup \hj$. By further using the condition of \stageTwoC\ that $(\xa,\xb) \lowerval{j} \xd \setminus\{\hj,\Hj\}$, we get that:
    \begin{equation}\label{p2: xc-gj hj}
    \big(\xa \cup \hj, \xb\big) \lowereqval{j} \xd \setminus \hj  \lowerval{j} (\xc \setminus \gj) \cup \hj.
    \end{equation}
    Therefore, it holds that $\xa  \lowerval{j} \xc \setminus \gj$, and by
    \cref{cancp} we get that:
    \begin{equation}\label{p2: xcp efxf for j}
        \xap = \xa \cup \lstitem{j}{1}{\{\hi, \Hi\}} \lowereqval{j}
        (\xc \setminus \gj) \cup \hstitem{j}{1}{\{\hi, \Hi\}} = \xcp.
    \end{equation}
    Moreover, by \cref{p2: xc-gj hj}, we get:
    \begin{equation}\label{p2: xcp efxf for j 2}
        \xbp = \xb \lowereqval{j}  (\xc \setminus \gj) \cup \hj \lowereqval{j}
        (\xc \setminus \gj) \cup \hstitem{j}{1}{\{\hi, \Hi\}}   = \xcp.
    \end{equation}
    Since $\xd \setminus \hj \efxenvyll{j} (\xc \cup \hj)$ (by \stageTwoC) 
    we get that $ \xd \setminus \hj \lowerval{j} (\xc \cup \hj) \lowereqval{j} (\xc \cup \Hj)$, which along with $\gj \lowereqval{j} \hj$ (by lemma's condition) results in:
    \begin{equation}
         \xd \setminus\{\hi, \Hi\}\lowereqval{j}  \xd \setminus \{\hj,\Hj\}  \lowereqval{j} \xc 
         \lowereqval{j}  (\xc \setminus \gj ) \cup \hj  \lowereqval{j} \xcp
    \end{equation}

    Since $\gj \lowereqval{j} \hj$ (by lemma's condition), we have that $\gj$ is the least valued \itm\ in 
    $\xdp$ w.r.t.\ agent $j$'s valuation, 
    so $\xcp \doesnotefxenvyll{j} \xdp$. This along with
    Equations \eqref{p2: xcp efxf for j} and \eqref{p2: xcp efxf for j 2}  gives
    \isefxffor{\xcp}{j}{\xp}.
     
     \donew
\end{proof}

\begin{lemma}
\label{DerivingPerfect3}
    Suppose $\X$ is a partition in \stageTwoC \ such that $\hj \lowerval{j} \gj$.
    Then we can construct a new partition $\xt$ that is either \eftxf \ or \good\
    with $\phi(\xt)>\phi(X)$ or in \stageOne \ with $\phi(\xt) = \phi(\X)$.
\end{lemma}
\begin{proof}
We will consider two cases regarding the preference of agent $i$.

    \noindent\textbf{$\bullet$ \caseA. [$\mathbf{\xc \greatereqval{i} \xa \cup \hj}$]:}
    We define the following allocation
    \begin{center}
        $\xp = \big(\xa \cup \hj,~~ \xb,~~ \xc,~~ \xd \setminus \hj \big)$. 
    \end{center}
    By the conditions of \StageTwoC it holds that \isbstffor{\xd}{j}{\X}, and hence 
     $\xd \setminus \hj \doesnotefxenvyll{j} (\xb,\xc)$,
    By the conditions of \StageTwoC it also holds that  $\xd \setminus \{\hj, \Hj\} \greaterval{j} \xa$, 
    that in turns gives $\xd \setminus \hj\greatereqval{j} \xa \cup \hj$. Overall \isefxffor{\xdp}{j}{\xp}.
    Regarding agent $i$, it holds that \isefxf{\xc}{i}{\X} (by \StageTwoC), which means that $\xc \doesnotefxenvyll{i} (\xb,\xd \setminus \hj)$.
    By the condition of this case, it overall holds that \isefxffor{\xcp}{i}{\xp}, 
    and by \cref{ToGoodPart from stage 2 case 1}, the lemma follows for this case.

    \vspace{5pt}
    \noindent\textbf{$\bullet$ \caseB. [$\mathbf{\xc \lowerval{i} \xa \cup \hj}$]:}
    In this case, we define the following allocation:
    \begin{equation}
        \xp = (\xa,~~ \xb,~~ \xc \cup \hj,~~ \xd \setminus \hj).
    \end{equation}
    
    \noindent Regarding agent $j$, by the fact that $\hj \lowerval{j} \gj$ (by lemma's assumption), 
    we get 
    $$(\xc \cup \hj)\setminus \{\gjp,\Gjp\} = \xc \setminus \gj \lowereqval{j} \xd \setminus \hj\,,$$
    where the last inequality is due to \isbstffor{\xd}{j}{\X} (by \stageTwoC).
    By the same condition, i.e.,  \isbstffor{X_4}{j}{\X} it turns that  
    \iseftxffor{\xdp}{j}{\xp}.
    We also have $\xcp = \xc \cup \hj \greaterval{j} \xd \setminus \hj = \xdp$ 
    (by \StageTwoC),
    also since \isbstffor{\xd}{i}{\X}, we have that $\xd \setminus \hj \doesnotefxenvyll{j} (\xa,\xb)$,
    therefore we get $\xcp \doesnotefxenvyll{j} (\xa,\xb)$, 
    so \iseftxffor{(\xcp,\xdp)}{j}{\xp}.
    
    Regarding agent $i$, since \isefxffor{\xc}{i}{\X} (by \stageTwoC), 
    for every $h \in \xd$ we get:
    $\xd \setminus h \lowereqval{i} \xc \lowerval{i} \xa \cup \hj.$
    Therefore, by cancelability, for $h = \hjp$, we get $$\xdp\setminus \hjp =\xd \setminus \{\hj, \hjp\} \lowerval{i} \xa \lowereqval{i} \argmax_i(\xa,\xb)\,.$$  
    Moreover, by the lemma's assumption, we have that $\hj \lowerval{j} \gj$, so we get $\hj \lowerval{j} \gi$, 
    and by \normalizationality \ of $(\xc,\xd)$ w.r.t.\ agents $(i,j)$, 
    we get  $\hj \lowereqval{i} \gi$.
    Therefore, $$\xcp\setminus \{\gip,\Gip\}=(\xc \cup \hj) \setminus \{\gip,\Gip\} = \xc\setminus \gi \lowerval{i} \argmax_i (\xa, X_2)\,,$$
    This means that $\argmax_i(\xap,\xbp)=$ \iseftxffor{\argmax_i(\xa,\xb)}{i}{\xp}; it also holds that \iseftxffor{\xcp}{i}{\xp}, 
    since \isefxffor{\xc}{i}{\X} (by \stageTwoC). 
     
    So far, we have that 
    \begin{align}
        \text{\iseftxffor{(\xcp,\xdp)}{j}{\xp} and \iseftxffor{(\argmax_i(\xap,\xbp),\xcp)}{i}{\xp}.} \label{eq:EF2XForij}
    \end{align} 
    Next, we turn our attention to agent $1$. If $\xa=$\isnotefxffor{\xap}{1}{\xp}, we show that the conditions of \cref{AgentAEnvies0} are satisfied for $\xp$.  Since \areefxffor{\xa,\xb}{1}{\X}, 
    we get $(\xa,\xb) \greatereqval{1} \xd \setminus \hj$.
    Therefore $\xa \efxenvyll{1} \xcp$, otherwise \isefxffor{\xap}{1}{\xp}. 
    Since $\xd \setminus \{\hj, \Hj\} \greaterval{j} \xa$, we get that $\xdp \setminus \hjp \greaterval{j} \xa$. 
    Also agent $j$ does not \eftxenvy\ bundle $\xcp$
    \with\ bundle $\xdp$ (by \ref{eq:EF2XForij}), 
    and does not \efxenvy\ bundle $\xb$, since \isbstffor{\xd}{j}{\X}. 
    Therefore, by \cref{AgentAEnvies0}
    we would complete the proof.

    Finally, suppose that \areefxffor{\xa}{1}{\xp}, which in turns means that \isefxffor{\xb}{1}{\xp}, 
    since $\xa= \argmin_1(X_1, X_2)$ (by \stageTwoC). 
    Hence, \areefxffor{\xap,\xbp}{1}{\xp}, and by combining this with  \cref{eq:EF2XForij}, we get that $\xp$ is \eftxf\ by \cref{bipartite graph}.
\end{proof}

\subsection{Dealing with Partitions in \StageTwoPerfect}
\label{stage 2d section}

    In this subsection we prove (\cref{solve perfect}) that if we start with a partition $\X$ in \stageTwoPerfect, 
    then we either construct an \eftxf\ partition or we construct another \goodpart \ 
    with higher potential. We first restate the conditions of \stageTwoPerfect. \\\\
    {\bf Conditions of \stageTwoPerfect} (Restated).
    {\em 
        We say partition $\X$ is in \stageTwoPerfect\ if all the following hold:
        \begin{itemize}
        \item $\xa= \argmin_1(X_1, X_2)$, and \areefxffor{\xa,\xb}{1}{\X};
        \item \isefxf{\xc}{i}{\X} for some agent $i \ne 1$, and $\xc \setminus \gi \greaterval{i} (\xa, \xb)$;
        \item $\xd \in \text{\bstfset}_j(\X)$ for some agent $j\notin \{1,i\}$, and $\xd \setminus \hj \greaterval{j} (\xa, \xb)$;
        \item $(\xc, \xd)$ are \normalized\  w.r.t.\ agents $i,j$ (see Definition~\ref{normalized sets def});
        \end{itemize}
        }

    Our main theorem of this section follows.

\begin{theorem}
\label{solve perfect}
    Given any partition $\X$ in \stageTwoPerfect, we can construct a new partition $\xt$ that is either \eftxf,
    or \good\ with $\phi(\xt)>\phi(\X)$.
\end{theorem}
\begin{proof}
    If there exists an agent among agents $i,j$ (as defined in the \stageTwoPerfect),  
    for whom at least one of the bundles \xab is \eftxf, by \cref{Cases} the theorem follows. So, by assuming that \xab are not \eftxf\ for agents $i,j$, we handle separately all the three possible cases described next:
\begin{enumerate}
    \item $\xd$ is \bsttwof \ for agents $i,j$.  

    \item $\xc$ is \bsttwof \ for agents $i,j$. 

    \item$\xc$ is \bsttwof \ for $k \in \{i,j\}$ and $\xd$ is \bsttwof \ for agent $v \in \{i,j\}, v \ne k$.  
\end{enumerate}
    We show the theorem for the first and second case, later in this subsection, in Lemmas \ref{Casei} and \ref{Caseii}, respectively. The theorem for the third case follows from \cref{SeperateBest-2-feasible}.
\end{proof}
    
    The following lemma is an auxiliary lemma for the proof of \cref{Cases} that handles the situation that in a partition in \stageTwoPerfect\ one of the bundles \xab is \eftxf\ for some agent different from $1$ (used in the proof of \cref{solve perfect}).
    
\begin{lemma}
\label{ToGoodPart8}
    Suppose $\X$ is a partition such that $\xa= \argmin_1(X_1, X_2)$ and for two distinct agents $i,j \ne 1$
    we have that \isefxffor{\xd}{j}{\X}, $(\xc,\xd)$ are \normalized\ w.r.t.\ agents $(i,j)$, and also the following hold:
    \begin{align}
    \xc \setminus \{\gi, \Gi\} &\greaterval{i}  (\xa, \xb, \xd) \label{eq_i_is_not_E2FX_with_1,2_and_4}\\
    \xd \setminus \hj &\greaterval{j}  (\xa, \xb) \label{eq_j_is_not_EFX_with_1_and_2}
    \end{align}
    Then, we can construct a new partition $\xt$ that is either \eftxf\ or \good\
    with $\phi(\xt)>\phi(\X)$.
\end{lemma}

\begin{proof}
    We define a new partition $\xp$ from $\X$ as follows:
        
    \begin{center}
    $\xp =  \big(\xa \cup \lstitem{j}{1}{\{\gi, \hj\}} , \hspace{2mm} \xb, 
            \hspace{2mm} \xc \setminus \gi, \hspace{2mm} (\xd \setminus \hj) \cup \hstitem{j}{1}{\{\gi, \hj\}} \big) $
    \end{center}
    We will show that \isefxffor{\xcp}{i}{\xp} and \isefxffor{\xdp}{j}{\xp}, so by \cref{ToGoodPart from stage 2 case 1} the lemma will follow. We start by showing that \isefxffor{\xdp}{j}{\xp}. For agent $j$ it holds that
    \begin{align*}
     \xdp  = (\xd \setminus \hj) \cup \hstitem{j}{1}{\{\gi, \hj\}} 
    &\greatereqval{j}   (\xd \setminus \hj) \cup \hj = \xd \greatereqval{j}    (\xb, \xcp) \,,
    \end{align*}
    where for the last preference it holds that $\xd \greatereqval{j} \xd \setminus \hj \greaterval{j}  \xb$ by \eqref{eq_j_is_not_EFX_with_1_and_2}, and since \isefxffor{\xd}{j}{\X}, $\xd \greatereqval{j} \xc \setminus \gj \greatereqval{j} \xc \setminus \gi = \xcp$. Moreover, by \eqref{eq_j_is_not_EFX_with_1_and_2} we also have that
    \begin{align*}
    \xdp  = (\xd \setminus \hj) \cup \hstitem{j}{1}{\{\gi, \hj\}} &\greatereqval{j} \xa \cup \hstitem{j}{1}{\{\gi, \hj\}}
    \greatereqval{j} \xa \cup \lstitem{j}{1}{\{\gi, \hj\}}   = \xap\,.
    \end{align*}
    Therefore \isefxffor{\xdp}{j}{\xp}. Next we show that it is also true that \isefxffor{\xcp}{i}{\xp}. For agent $i$ by \eqref{eq_i_is_not_E2FX_with_1,2_and_4} it holds that
    \begin{align*}
     \xcp = \xc \setminus \gi &\greatereqval{i} \xc \setminus \{\gi, \Gi\} \greaterval{i} \xb = \xbp \ \mbox{and}\\
    \xcp =  \xc \setminus \gi &\greatereqval{i}  X_r \cup \Gi \greatereqval{i} X_r \cup \gi \,, \mbox{ for any $r\in \{1,4\}$.} 
    \end{align*}
    If $ \lstitem{j}{1}{\{\gi, \hj\}} = \gi$ (and therefore $ \hstitem{j}{1}{\{\gi, \hj\}} = \hj$), we trivially get that 
    
    \begin{center}
       $\xcp \greatereqval{i}\xa \cup \gi = \xap$, and $\xcp \greatereqval{i}\xd  = \xdp$. 
    \end{center}
    If $ \lstitem{j}{1}{\{\gi, \hj\}} = \hj$ (and therefore $ \hstitem{j}{1}{\{\gi, \hj\}} = \gi$), then $ \hj \lowereqval{j} \gi$, 
    so by \normalizationprocess \ of $(\xc,\xd)$ w.r.t.\ agents $(i,j)$,
    we get  that $ \hj \lowereqval{i} \gi$, and so
    
    \begin{center}
       $\xcp \greatereqval{i} \xa \cup \hj = \xap$, and $\xcp \greatereqval{i}\xd \cup \gi \greatereqval{i}(\xd \setminus \hj) \cup \gi  = \xdp$. 
    \end{center}
    So, it also holds that \isefxffor{\xcp}{i}{\xp}. 
\donew
\end{proof}

\begin{lemma}\label{Cases}
    Suppose  that $\X$ is in \stageTwoPerfect\ and there exists an agent
    among agents $(i,j)$ (as defined in the \stageTwoPerfect), 
    for whom at least one of the bundles \xab is \eftxf. 
    Then we can construct a new partition $\xt$ that is either
    \eftxf\ or \good\ with $\phi(\xt)>\phi(\X)$.
\end{lemma}

\begin{proof}
    We first consider the case that for either agent $i$ or agent $j$ (as defined in the \stageTwoPerfect) there does not exist at least two bundles
    that are \eftxf\ for them. 
    Suppose that this is agent $i$, 
    then since \isefxffor{\xc}{i}{\X} (by Definition of \stageTwoPerfect), $\xc$ is the only \eftxf\ for agent $i$, and therefore, 
     $\xc \setminus \{\hi, \Hi\} \greaterval{i} (\xa, \xb, \xd)$ (by \cref{{best props}}).
    Moreover, by Definition of \stageTwoPerfect, it holds that \isefxffor{\xd}{j}{\X} and $\xd \setminus \hj \greaterval{j} (\xa, \xb)$. 
    Then we use \cref{ToGoodPart8}, 
    for the construction of a new partition $\xt$ that is either \eftxf\ or \good\
    with $\phi(\xt)>\phi(\X)$.
    The same is applicable for agent $j$ since we did not use the fact that $\xd$ is \bstf \ for agent $j$.
    
     The case left is the case where for each of the agents $i$ and $j$, there are at least two bundles that are \eftxf\ for them.
     Then, by \cref{bipartite  graph} $\X$ is \eftxf.
\end{proof}

    The next lemma handles the first case of the proof of \cref{solve perfect}.
    
\begin{lemma}
\label{Casei}
    If partition $\X$ is in \stageTwoPerfect, the bundles \xab are not \eftxf\ for any of the agents but $1$, and it holds that $\xd$ is \bsttwof \ for agents $i,j$, as defined for \stageTwoPerfect, (case 1 of \cref{solve perfect}), we can construct a new partition $\xt$ that is either \eftxf\ 
    or \good\ with $\phi(\xt)>\phi(\X)$.
\end{lemma}

\begin{proof}
    For agent $j$, \isbsttwoffor{\xd}{j}{\X}  by lemma's statement and, since \xab are not \eftxf\ for agent $j$, it holds that \arenotbsttwoffor{\text{\xab}}{j}{\X}. By \cref{{best props}} it then holds that
    \begin{align}
      \xd \setminus \{\hj, \Hj\} \greaterval{j} \text{\xab} \mbox{, and so, } \xd \setminus \hj \greatereqval{j}   X_r \cup \Hj  \greatereqval{j}   X_r \cup \hj \mbox{ for } r\in \{1,2\} \label{eq:jconditions}  
    \end{align} 
    By using the same reasoning for agent $i$, it holds that 
    \begin{align}
    \xd \setminus \{\hi, \Hi\} \greaterval{i} \text{\xab} \mbox{, and so, } \xd \setminus \hi \greatereqval{i}   X_r \cup \Hi  \greatereqval{i}   X_r \cup \hi \mbox{ for } r\in \{1,2\} \label{eq:iconditions}
    \end{align}
    Moreover, by the conditions of \stageTwoPerfect\ it holds that $\xc \setminus \gi \greaterval{i} \xa$, 
    and so it holds that 
    \begin{align}
    \xc \greatereqval{i}  \xa \cup \gi. \label{eq:iprefersX3}
    \end{align}

    We next consider three sub-cases based on how agent $i$ values the \itms\ $\gi, \Hi,$ and $\hj$.

    \textbf{\subcaseA.  $\mathbf{\gi \lowerval{i} \Hi}$:}
    The lemma follows by using \cref{DerivingPerfect1}, whose requirements are held by Equations \eqref{eq:jconditions} and \ref{eq:iconditions}.
    
    \textbf{\subcaseB. $\mathbf{\gi \greatereqval{i} \hj}$:}
    Consider the new partition 
    $$\xp = \big( \xa \cup \hj, \hspace{2mm}  \xb, \hspace{2mm} \xc,  \hspace{2mm} \xd \setminus \hj \big)\,.$$ 
    Regarding agent $j$, by \cref{eq:jconditions}, we have that $\xd \setminus \hj \greatereqval{j}  \xa \cup \hj$. Moreover, by the conditions of \stageTwoPerfect, it is \isbstffor{\xd}{j}{\X}, and so by \cref{best-feasible condition} it holds that $ \xd \setminus \hj \doesnotefxenvyll{j} (\xb,\xc)$. 
    Overall, \isefxffor{\xdp}{j}{\xp}.
    
    Regarding agent $i$, \isefxffor{\xc}{i}{\X} (by the conditions of \stageTwoPerfect), 
    therefore it trivially holds that $\xc \doesnotefxenvyll{i} (\xb, \xd \setminus \hj)$. Also, by using \cref{eq:iprefersX3} and the condition of this sub-case, $\xc \greatereqval{i}  \xa \cup \hj$.
    Hence, \isefxffor{\xcp}{i}{\xp}.
    
    \donew

    \textbf{\subcaseC. $\mathbf{\Hi \lowereqval{i} \gi \lowerval{i} \hj}$:} Consider the new partition
    \begin{center}
        $\xp = \big( \xa \cup \gi, \hspace{2mm}  \xb, \hspace{2mm} (\xc \setminus \gi) \cup \hj,  \hspace{2mm} \xd \setminus \hj \big)$
    \end{center}
    We first show that \isefxffor{\xcp}{i}{\xp}. 
    By using the condition of this sub-case $\xcp \greatereqval{i} \xc$, 
    and  \isefxffor{\xc}{i}{\X} (by the conditions of \stageTwoPerfect), 
    it trivially holds that $\xcp \doesnotefxenvyll{i} (\xb, \xd \setminus \hj )$.  
    By \cref{eq:iprefersX3}, it further holds that $\xcp \greatereqval{i} \xc \greatereqval{i} \xa \cup \gi$ 
    Hence, \isefxffor{\xcp}{i}{\xp}.

    Finally, for agent $j$ we either show that \isefxffor{\xdp}{j}{\xp} or we will consider another partition. By the conditions of \stageTwoPerfect, it holds that  \isbstffor{\xd}{j}{\X}, 
    and so by \cref{best-feasible condition} it holds that $\xd\setminus\hj \doesnotefxenvyll{j} \xb$. 
    By \normalizationality\ of $(\xc,\xd)$ w.r.t.\ agents $(i,j)$, 
    since $ \gi \lowerval{i} \hj$,  we get  
    \begin{align}
    \gi \lowereqval{j} \hj , \label{eq:inferbyswpa-opt} 
     \end{align}
    and by using also \cref{eq:jconditions} we get that:
    \begin{align*}
        \xd \setminus \hj &\greatereqval{j} \xa \cup \hj \greatereqval{j} \xa \cup \gi. 
    \end{align*}
    If $ \xdp=\xd \setminus \hj \doesnotefxenvyll{j} \xcp$, then \isefxffor{\xdp}{j}{\xp}, and the lemma would follow by the use of \cref{ToGoodPart from stage 2 case 1}. So, for the rest of the proof, we assume otherwise, i.e., $\xd \setminus \hj \efxenvyll{j} \xcp$. This would mean that
    \begin{align*} 
        \xd \setminus \hj &\lowerval{j} \xcp \setminus \gjp =
        \big( (\xc \setminus \gi) \cup \hj \big) \setminus \gjp. 
    \end{align*}
    If $\hj = \gjp$, then $\xd \setminus \hj \lowerval{j}  \xc \setminus \gi$, which contradicts the fact that \isbstffor{\xd}{j}{\X} (condition of \stageTwoPerfect). Hence, $\hj \ne \gjp$, so $ \Gjp \lowereqval{j} \hj \lowereqval{j} \Hj$ and by using the above preference, $\xd \setminus \hj \lowerval{j} \xcp \setminus \gjp$, we get:
    \begin{align}
       \xd \setminus \{\hj, \Hj\}  &\lowerval{j}   \xcp \setminus \{\gjp,\Gjp\}. \label{eq:jEF2X_compare_X3_X4}
    \end{align}
    
    We define the following partition, for which we will show that the conditions of 
    \cref{SeperateBest-2-feasible} are satisfied by swapping agents $i$ and $j$:  
    $$\xz = \big( \xa, \hspace{2mm}  \xb, \hspace{2mm} (\xc \setminus \gj) \cup \hi,  \hspace{2mm} (\xd \setminus \hi)  \cup \gj \big)$$
    Note that $\xcz=  (\xc \setminus \gj) \cup \hi$ can be derived  from 
    $\xcp = (\xc \setminus \gi) \cup \hj$ by swapping item $\hj$ with $\hi$ and item $\gj$ with $\gi$.
    Since we have $\hj \lowereqval{j} \hi$ and $\gj \lowereqval{j} \gi$, 
    by using \cref{best-k-value prop} twice, 
    we get $\xcz \setminus \{\gjz,\Gjz\} \greatereqval{j}  \xcp \setminus \{\gjp,\Gjp\}$. 
    Hence by using also Equations \eqref{eq:jEF2X_compare_X3_X4},  \eqref{eq:jconditions} and \eqref{eq:iconditions}: 
    \begin{align*}    
        \xcz \setminus \{\gjz,\Gjz\} &\greatereqval{j}  \xcp \setminus \{\gjp,\Gjp\}\greaterval{j}
        \xd \setminus  \{\hj, \Hj\} \greaterval{j} (\xa, \xb) \\
        \xdz \setminus \{\hiz,\Hiz\} &\greatereqval{i}  \xdz \setminus \{\gj, \Hi\} =\hspace{1mm}
        \xd \setminus  \{\hi, \Hi\}  \greaterval{i} (\xa, \xb).
    \end{align*}
 By \cref{eq:inferbyswpa-opt}, it holds that $\gi \lowereqval{j} \hj$, therefore $\gj \lowereqval{j} \gi \lowereqval{j} \hj$, 
    so we get that $\lst{j}{1}{\xdz}=\gj$, and by using also \cref{eq:jEF2X_compare_X3_X4}:
    
    \begin{center}
        $\xdz \setminus \{\hjz,\Hjz\} \lowereqval{j} \xd \setminus \{\hj, \Hj\}  \lowerval{j}    
         \xcp \setminus \{\gjp,\Gjp\} \lowereqval{j}  \xcz \setminus \{\gjz,\Gjz\} $.
    \end{center}
    
    Finally, since $\hi \lowereqval{i} \Hi \lowereqval{i} \gi$ we get that $\lst{i}{1}{\xcz}=\hi$, and by further using the fact that \isbsttwoffor{\xd}{i}{\X}:
    
    \begin{center}
        $ \xcz \setminus \{\giz,\Giz\} \lowereqval{i}    \xc \setminus \{\gi, \Gi\} \lowereqval{i}   \xd \setminus \{\hi, \Hi\}  \lowereqval{i}  \xdz \setminus (\hiz, \hiz)$.
    \end{center}

    So, we have:
    \begin{align*}
        \xcz \setminus \{\gjz,\Gjz\}  &\greaterval{j}  (\xdz \setminus \{\hjz,\Hjz\}, \xa, \xb)\\
        \xdz \setminus \{\hiz,\Hiz\}  &\greaterval{i}  (\xcz \setminus \{\giz,\Giz\}, \xa, \xb). 
    \end{align*}
    Therefore, by \cref{SeperateBest-2-feasible}, the lemma follows.
\end{proof}

    The following lemma is an auxiliary lemma for the proof of \cref{Caseii}
    that handles the situation that will occur in different cases.

\begin{lemma}
\label{auxiliary lemma case ii}
    Suppose $\X$ is a partition such that $\xa\lowerval{1}\xb$, \isbstffor{\xd}{j}{\X}, and the following hold: 
    \begin{align}
        \Gj &\lowereqval{j} \hj  \label{hj > Gj} \\  
        (\xa,\xb)&\lowereqval{j}  \xc \setminus \{\gj,\Gj\} \label{xc-->} \\  
        \xd\setminus\hi &\lowereqval{j}  \xc \setminus \{\gj,\Gj\} \cup \hi \label{xd-<xc--+}\\
        (\xa,\xb)&\lowereqval{i} \xd \setminus \hi  \label{xa xb <i xd-hi}
    \end{align}
    Then we can construct a new partition $\xt$ that is either \eftxf\ 
    or \good\ with $\phi(\xt)>\phi(\X)$.
\end{lemma}
\begin{proof}
    By \cref{hj > Gj} we get:
    \begin{equation}\label{hj > gj Gj}
        \gj\lowereqval{j}\Gj \lowereqval{j} \hj.
    \end{equation}
    
    Define partition $\xp$ as follows:
    \begin{center}
        $\xp= \Big( 
        \xa \cup \bigerlstitem{i}{1}{\{\gj, \Gj\}}, \hspace{2mm}  \xb, \hspace{2mm}
        (\xc \setminus \{\gj, \Gj\}) \cup \hi ,  \hspace{2mm} 
        (\xd \setminus \hi) \cup \bigerhstitem{i}{1}{\{\gj, \Gj\}} \Big)$
    \end{center}

    \noindent We first prove that \isefxffor{\xc}{j}{\xp}.
    By \cref{xc-->} we get 
    \begin{equation} \label{case ii no efx-envy for j2}
        (\xa, \xbz) \lowerval{j} \xc \setminus \{\gj, \Gj\},      
    \end{equation}
    and  by \cref{hj > gj Gj}, since $\lstitem{i}{1}{\{\gj, \Gj\}} \in \{\gj, \Gj\}$ we get
    \begin{equation} \label{lst <j hj}
         \lstitem{i}{1}{\{\gj, \Gj\}} \lowereqval{j} \hj  \lowereqval{j} \hi.
    \end{equation}
    So by \cref{case ii no efx-envy for j2}, \cref{lst <j hj} and \cref{cancp}:
    \begin{equation}\label{case ii no efx-envy for j1}
        \xap = \xa \cup \lstitem{i}{1}{\{\gj,\Gj\}} \lowereqval{j} \xc \setminus \{\gj, \Gj\}  \cup \hi =\xcp.
    \end{equation}
    
     By \eqref{hj > gj Gj}, agent $j$ prefers any good in $\xd$ over either one of her two least valuable goods in $\xc$. In other words, $j$ values any good in $\xd$ (and hence also in $\xd\setminus \hi$) at least as much as any good in $\{\gj, \Gj\}$, implying that $j$'s least valuable item in $\xdz$ is the one we moved from $\xc$ to $\xd \setminus \hj$, i.e.,:
    \begin{equation} \label{least gj or Gj}
    \lstitem{j}{1}{\xdp} =\biglstitem{j}{1}{(\xd \setminus \hi) \cup \hstitem{i}{1}{\{\gj, \Gj\}}}
    = \hstitem{i}{1}{\{\gj, \Gj\}}.  
    \end{equation}       
    Also by \cref{xd-<xc--+} we get
    \begin{equation*}
        \xd \setminus \hi \lowereqval{j}  
        \xc \setminus \{\gj, \Gj\}  \cup  \hi = \xcp,
    \end{equation*}
    so agent $j$ weakly prefers $\xcp$ over $\xd \setminus \hi$ and, 
    since \cref{least gj or Gj} shows that her least valuable item in $\xdp$ 
    is the one that was moved from  $\xc$ to $\xd \setminus \hi$, we can conclude that
    \begin{equation}\label{case ii no efx-envy for j3}
     \xcp \doesnotefxenvyll{j}
         (\xd \setminus \hi) \cup \hstitem{i}{1}{\{\gj,\Gj\}} =\xdp. 
    \end{equation}

    \noindent Hence, by \cref{case ii no efx-envy for j2}, \cref{case ii no efx-envy for j1}, 
    and \cref{case ii no efx-envy for j3} we get \isefxffor{\xcz}{j}{\xz}.
    By \cref{xa xb <i xd-hi}, we have that
    \begin{equation}
        \xbp \lowereqval{i} \xdp.
    \end{equation}
    Also by using the fact $\lstitem{i}{1}{\{\gj,\Gj\}} \lowereqval{i} \hstitem{i}{1}{\{\gj,\Gj\}}$,
    \cref{xa xb <i xd-hi}, and \cref{cancp} we get
    \begin{equation}
        \xap  = \xa \cup \lstitem{i}{1}{\{\gj,\Gj\}}  \lowereqval{i}
        (\xd \setminus \hi) \cup \bigerhstitem{i}{1}{\{\gj, \Gj\}} = \xdp.
    \end{equation}
    \noindent If \isefxffor{\xdp}{i}{\xp}, the lemma follows by \cref{ToGoodPart from stage 2 case 1}. So assume otherwise, i,e.,:
    \begin{equation}
        (\xap, \xbp) \lowerval{i} \xdp \lowerval{i} \xcp.
\end{equation}
    Define $\yc, \yd$ as follows:
    \begin{align*}
        \yc &= (\xc \setminus \{\gj, \Gj\} ) \cup  \hj  \\
        \yd &=   (\xd \setminus \hj) \cup \hst{i}{1}{\{\gj, \Gj\}} 
    \end{align*}    
    We have that $\yc \sqcup \yd = \xap \sqcup \xbp $.
    By \cref{xc-->}, \cref{hj > gj Gj}, and \cref{cancp} we get:
    \begin{equation}
        \argmax_j(\xap, \xbp) \lowereqval{j} \argmax_j(\xa, \xb)  \cup \bigerlstitem{i}{1}{\{\gj, \Gj\}}
        \lowereqval{j} (\xc \setminus \{\gj, \Gj\} ) \cup  \hj  = \yc.
    \end{equation}
    Also since \isbstffor{\xd}{j}{\X} and by \cref{xc-->} we get:
    \begin{equation}
        \xa \cup \bigerlstitem{i}{1}{\{\gj, \Gj\}} \lowereqval{j}
        \xa \cup \Gj  \lowereqval{j}  \xc \setminus \gj 
        \lowereqval{j} \xd \setminus \hj  \lowereqval{j} \yd. 
    \end{equation}
    Also since \isbstffor{\xd}{j}{\X} and by \cref{xc-->} we get:
    \begin{equation}
        \xbp = \xb \lowereqval{j}  \xc \setminus \gj \lowereqval{j} 
        \xd \setminus \hj \lowereqval{j} \yd.
    \end{equation}    
    Therefore we have that $(\yc,\yd) \greaterval{j} (\xap,\xbp)$, $(\xcp,\xdp) \greaterval{i} (\xap,\xbp)$,
    and $\yc \sqcup \yd = \xap \sqcup \xbp $,
    so by \cref{ToGoodPart5}, 
    we can construct a new partition $\xt$ that is either \eftxf\ or \good\ 
    with $\phi(\xt)\geq \mintwop >v_1(\xa) =  \phi(\X)$.
\end{proof}
    
    The next lemma handles the second case of the proof of \cref{solve perfect}.
    
\begin{lemma}
\label{Caseii}
    If partition $\X$ is in \stageTwoPerfect, the bundles \xab are not \eftxf\ for any of the agents but $1$, and it holds that $\xc$ is \bsttwof\ for agents $i,j$, as defined for \stageTwoPerfect, (case 2 of \cref{solve perfect}), we can construct a new partition $\xt$ that is either \eftxf\ 
    or \good\ with $\phi(\xt)>\phi(\X)$.
\end{lemma}

\begin{proof}
    If $\xc \in \text{\bstfset}_i(\X)$, 
    then by renaming bundle $\xc$ to $\xd$, 
    and bundle $\xd$ to $\xc$, the lemma follows by \cref{Casei}. 
    So, we will assume otherwise, i.e., $\xc \notin \text{\bstfset}_i(\X)$. 
    Additionally, since none of the bundles \xab are \eftxf\ (and therefore \efxf) 
    for agent $i$ in $\X$, 
    by \cref{best props}, we get \xab $\notin \text{\bstfset}_i(\X)$. 
    Therefore, \isbstffor{\xd}{i}{\X}. 
    To summarize, 
    \begin{align}
    \text{$\xc$ is \bsttwof \ and $\xd$ is \bstf\ for both agents $i,j$\,,} \label{eq:EFXf-EF2XfCaseii}
    \end{align}
    and moreover, we have:
    \begin{equation} \label{xd-hi >i xa xb}
        (\xa,\xb) \lowereqval{i} \xc \setminus \gi \lowereqval{i} \xd \setminus \hi.
    \end{equation}
    Additionally, since none of the bundles \xab are \eftxf\ for agents $i$ and $j$, it holds that
    \begin{align}
     (\xa,\xb) \lowerval{i}  \xc \setminus\{\lstitem{i}{1}{\xc},\lstitem{i}{2}{\xc}\} \mbox{ and }   (\xa,\xb) \lowerval{j}  \xc \setminus\{\lstitem{j}{1}{\xc},\lstitem{j}{2}{\xc}\}. \label{eq:preferencesCaseii} 
    \end{align}

    We next consider the following two cases based on how agent $j$ values the \itms\ $\gi$ and $\hj$. 

    \vspace{10pt}
    \noindent\textbf{$\bullet$ Sub-case 1 [$\mathbf{ \gi \greaterval{j} \hj}$]:} 
    By \normalizationality\ of $(\xc,\xd)$ w.r.t.\ agents $(i,j)$,
    it further holds that 
    \begin{align} 
    \gi \greatereqval{i} \hj. \label{eq:normalizationCaseii} 
    \end{align} 

    If $\xa \cup \hj \lowereqval{j} \xd \setminus \hj$, we define the allocation:
    \begin{center}
    $\xp = \big( \xa \cup \hj, \hspace{2mm}  \xb, \hspace{2mm} \xc ,  \hspace{2mm} \xd \setminus \hj \big)$
    \end{center}
    By $\xa \cup \hj \lowereqval{j} \xd \setminus \hj$  and the fact that \isbstffor{\xd}{j}{\X} (by \eqref{eq:EFXf-EF2XfCaseii}), we get that \isefxffor{\xdp}{j}{\xp}. Regarding agent $i$, by \cref{eq:preferencesCaseii} and the fact that 
    $\gi \greatereqval{i} \hj$ (by \eqref{eq:normalizationCaseii}), we have that:
    
    \begin{center}
    $ \xc \greatereqval{i} \xc \setminus \gi \greatereqval{i} \xa \cup \Gi \greatereqval{i} 
    \xa \cup \gi \greatereqval{i} \xa \cup \hj$.
    \end{center}
    Additionally, by \cref{eq:preferencesCaseii} it holds that $\xc \greatereqval{i} \xb$, and by the conditions of \stageTwoPerfect\ it holds that \isefxffor{\xc}{i}{\X}, meaning that $\xc \greatereqval{i} \xd \setminus \hi \greatereqval{i} \xd \setminus \hj$. Overall, \isefxffor{\xcp}{i}{\xp}, and the lemma follows by \cref{ToGoodPart from stage 2 case 1}.
    
    We now consider the other case where:
    \begin{align}
        \xd \setminus \hj \lowerval{j} \xa \cup \hj \label{xd-hj<xa+hj}.
    \end{align}
    And we define instead the allocation:
    
    \begin{center}
    $\xp= \big( \xa \cup \hi,  \hspace{2mm}  \xb, \hspace{2mm} \xc ,  \hspace{2mm} \xd \setminus \hi \big)$
    \end{center}

    We first show that \isefxffor{\xdp}{i}{\xp}.
    Since \isbsttwoffor{\xc}{i}{\X}, \isbstffor{\xd}{i}{\X} (by \eqref{eq:EFXf-EF2XfCaseii}), and \isnotbstffor{\xc}{i}{\X} (by our assumption at the beginning of the proof) we get:
    \begin{align}
        \xc \setminus\{\gi, \Gi\} & \greatereqval{i} \xd \setminus\{\hi, \Hi\}, \notag\\
        \xc \setminus \gi & \lowerval{i} \xd \setminus \hi, \label{eq:X3notEFXBestFori}
    \end{align}
    therefore by \cref{cancp}, we get $\gi \lowereqval{i} \Gi \lowerval{i} \Hi$. By \cref{eq:normalizationCaseii} it therefore holds that $\hj \lowereqval{i} \gi \lowerval{i} \Hi$, and so, it should be that:
    \begin{align}
        \hj = \hi. \label{hiequalshj}    
    \end{align}
     We put the above together, and so by using the Equations \eqref{hiequalshj}, \eqref{eq:normalizationCaseii},  \eqref{eq:preferencesCaseii} and \eqref{eq:X3notEFXBestFori}, we get 
    
    \begin{center}
    $ \xa \cup \hi = \xa \cup \hj \lowereqval{i} \xa \cup \gi \lowereqval{i} \xc  \setminus \gi \lowerval{i}  \xd  \setminus \hi$.
    \end{center}
    With this the additional fact that \isbstffor{\xd}{i}{\X} (by \eqref{eq:EFXf-EF2XfCaseii}),
    results in \isefxffor{\xdp}{i}{\xp}.

    If it was also true that \isefxf{\xcp}{j}{\xp}, the lemma would follow for this subcase by \cref{ToGoodPart from stage 2 case 1}. We next show that $\xc \greatereqval{j} (\xb,\xdp \setminus \hjp)$ by using \cref{eq:preferencesCaseii} and the fact that \isbsttwoffor{\xc}{j}{\X} (by \eqref{eq:EFXf-EF2XfCaseii}): 
    \begin{align*}
    \xc &\greatereqval{j} \xc \setminus \{\gj, \Gj\}\\ 
        &\greatereqval{j} (\xb,\xd \setminus \{\hj, \Hj\})\\ 
        &\greatereqval{j} (\xb,\xd \setminus \{\hi, \hjp\})\\ 
        &\greatereqval{j} (\xb,\xdp \setminus \hjp)\,,
    \end{align*}
    
    and if it was also true that $\xc \greatereqval{j} \xap= \xa \cup \hi$ then \isefxf{\xcp}{j}{\xp}, and the lemma would follow for this subcase. So, we assume otherwise, i.e.,
    \begin{align}
    \xc \lowerval{j} \xa \cup \hi\,,      \label{eq:X3loewerxahi}
    \end{align}
and we will show that for this case, the conditions of \cref{auxiliary lemma case ii} are satisfied, and the lemma follows.
    
    By Equations \eqref{eq:preferencesCaseii} and \eqref{eq:X3loewerxahi}, we get that:
    \begin{align*}
        \xa \cup \Gj \lowereqval{j} \xc \setminus \gj &\lowereqval{j} \xc \lowerval{j} \xa \cup \hi\,,
    \end{align*}
    which in turn gives, by using also \cref{hiequalshj}:
    \begin{equation}\label{hj>gj Gj}
     \Gj \lowereqval{j} \hi = \hj,   
    \end{equation}

    \noindent Finally, by Equations \eqref{xd-hj<xa+hj}, \eqref{eq:preferencesCaseii},
    and \cref{hiequalshj} we can infer that:
    \begin{align}    
         \xd \setminus \hi   \lowereqval{j} \xd \setminus \hj   \lowerval{j} \xa \cup \hj &\lowereqval{j} 
        \xc \setminus \{\gj, \Gj\}  \cup  \hj     \notag     \\
        &=\hspace{1.9mm} \xc \setminus \{\gj, \Gj\}  \cup  \hi.    \label{aux a 4}
    \end{align}

    \noindent Overall, Equations \eqref{hj>gj Gj}, \eqref{eq:preferencesCaseii}, \eqref{aux a 4}, \eqref{xd-hi >i xa xb},
    and the fact that $\X$ is in \stageTwoPerfect,
    provide the conditions of \cref{auxiliary lemma case ii}'s statement for $\X$, which concludes the proof for sub-case 1.

    \vspace{10pt}
    \noindent\textbf{$\bullet$ Sub-case 2. [$\mathbf{ \gi \lowereqval{j} \hj}$]:} Let
    
    \begin{center}
    $\Y = \big(\xa \cup \gi, \hspace{2mm}  \xb, \hspace{2mm} \xc \setminus \gi,  \hspace{2mm} \xd \setminus \hj \big)$
    \end{center}
    with unallocated set of \itms\ $\unassignedItms = \{\hj\} $. 
    By \cref{eq:preferencesCaseii}, we get 
    \begin{equation*}
        (\xb, \xa \cup \gi) \lowereqval{i} \xc \setminus \gi, 
    \end{equation*}
    and by \cref{eq:EFXf-EF2XfCaseii}, we get
    \begin{equation} 
        \yd \setminus \hiy = \xd \setminus \{\hj, \hiy\} \lowereqval{i}
        \xd \setminus \{\hi, \Hi\} \lowereqval{i}  \xc \setminus \gi\,. \notag
    \end{equation}
    Overall, it holds that \isefxffor{\yc}{i}{\Y}. 
    Because $\X$ is in \stageTwoPerfect\ and  
    \isbstffor{\xd}{j}{\X} (by \eqref{eq:EFXf-EF2XfCaseii}):
    \begin{center}
        $(\xb, \xc \setminus \gi) \lowereqval{j} \xd \setminus \hj\,.$
    \end{center}
    If additionally $\xd \setminus \hj \greatereqval{j} \xa \cup \gi$, it would be that \isefxffor{\yd}{j}{\Y}, 
    and the lemma would follow by \cref{ToGoodPart from stage 2 case 1}. So, we assume otherwise, i.e., that:
    \begin{align}
        \xd \setminus \hj \lowerval{j} \xa \cup \gi\,,    \label{xd-hj <j xa+gi}    
    \end{align}
    and we will show that for this case, the conditions of \cref{auxiliary lemma case ii} are satisfied, and the lemma follows.
    
    By Equations \eqref{eq:preferencesCaseii}, \eqref{eq:EFXf-EF2XfCaseii} and \eqref{xd-hj <j xa+gi}:
    \begin{align*}
        \xa \cup \Gj  \lowereqval{j} \xc \setminus \gj &\lowereqval{j}  
        \xd \setminus \hj  \lowerval{j} \xa \cup \gi\,,
    \end{align*}
    which in turn gives, by using also the condition in sub-case 2: 
    \begin{align}
        \Gj \lowerval{j} \gi &\lowereqval{j} \hj \lowereqval{j} \hi\,.   \label{< order fot j}
    \end{align}
    Finally, by Equations \eqref{xd-hj <j xa+gi}, and \eqref{< order fot j}
    we get:
    \begin{equation}\label{aux 4}
        \xd \setminus \hj \lowerval{j} \xa \cup \gi \lowereqval{j}   
        \xa \cup \hi  \lowereqval{j}   \xc \setminus \{\gj, \Gj\} \cup \hi.          
    \end{equation}
    Overall, Equations \eqref{< order fot j}, \eqref{eq:preferencesCaseii}, \eqref{aux 4}, \eqref{xd-hi >i xa xb},
    and the
    fact that $\X$ is in \stageTwoPerfect,
    provide the conditions of \cref{auxiliary lemma case ii}'s statement for $\X$, which concludes the proof for sub-case 2, as well.
\end{proof}

\subsection{Running time of Algorithm for 4 Agents}\label{running_time_4agents}
    Note that throughout the algorithm agent $1$, who has a nondegenerate valuation, is fixed. On the other hand, the remaining three agents, $\{2,3,4\}$, are not fixed, and we instead use indices $i,j,u$ to refer to them and their preferences.

\begin{theorem}
    The proposed algorithm returns an \eftx\ allocation in pseudo-polynomial time.
\end{theorem}
\begin{proof}
    Every loop in the flowchart includes the \goodpart \ node, 
    and every arrow to \goodpart \ node is accompanied by a strict increase of the potential. 
    So, since every cycle can only contain arrows without lowering potential, every cycle brings about 
    a strict increase of the potential. Since there are finitely many bundles (in fact, $2^{|M|}$), there are only finitely many
    possible values of the potential of each partition. Hence, there could be only finitely many cycles in the flowchart.
    So, our algorithm will eventually terminate, and therefore it will find an \efxf\ partition.
    
    Next, we investigate its running time.
    We first argue that every PR algorithm runs in pseudo-polynomial time.
    This is because after at most $m\cdot n$ steps in PR algorithm, 
    we have a strict increase in the minimum value bundle.
    Hence if we let  $\delta = \min_{S,T\subseteq M: v(S)\ne v(T)} |v(S)-v(T)|$ where $v$ is the input valuation function,
    then we  get that there could be at most $\frac{m\cdot n\cdot v(M)}{\delta}$ step in PR algorithm.
    Hence, PR algorithm runs in pseudo polynomial time.
    Next, we show that the whole algorithm runs in pseudo-polynomial time.
    Let $\Delta = \min_{S,T\subseteq M: v_1(S)\ne v_1(T)} |v_1(S)-v_1(T)|$. 
    Then the number of cycles is polynomial in $\frac{v_1(M)}{\Delta}$;
    hence, the number of cycles is pseudo-polynomial.
    It is not hard to verify that, except running time of PR algorithm, 
    all other operations in every edge run in polynomial time.
    Also, every PR algorithm runs in pseudo-polynomial time,
    and we will have a polynomial number of executions of PR algorithm at every edge of \cref{fig:algorithm-flow}, 
    therefore, the whole algorithm runs in pseudo-polynomial time.
\end{proof}

 \section{\eftx\ Allocations for Three Agents in Polynomial Time}
\label{sec:3agents}

    For instances involving 3 agents with cancelable valuations, our procedure from the previous section can be adapted to compute EF2X allocation in polynomial time.

    \vspace{0.1in}
    \begin{theorem}
        For every instance involving three agents with cancelable valuations and any number of goods, we can compute an EF2X allocation in polynomial time.   
    \end{theorem}
    
    We first define a similar \stageTwo\ for three agents.

    \begin{definition}[Stage B for three agents]
        We say partition $\X=(\xa,\xb,\xc)$ is in \stageTwo\ if $\xb$ is \efxf\ for some agent $i$
        and $\xc$ is \efxf\ for some other agent $j$.
    \end{definition}

    \begin{definition}[Stage B1 for three agents]
        We say partition $\X=(\xa,\xb,\xc)$ is in \stageTwoA\ if it is in \stageTwo,
        $\xc$ is \bstf\ for $j$, and:
        \begin{align*}
            \xc \setminus \gj \greaterval{j} \xa
        \end{align*}
    \end{definition}

    \begin{definition}[Stage B2 for three agents]
        We say partition $\X=(\xa,\xb,\xc)$ is in \stageTwoNormal\ if it is in \stageTwoA\ and
        $(\xb,\xc)$ are \normalized\ w.r.t.\ agents $(i,j)$.
    \end{definition}

    \begin{definition}[Stage B2i for three agents]
        We say partition $\X=(\xa,\xb,\xc)$ is in \stageTwoC\ if it is in \stageTwoNormal\ and:
        \begin{align*}
        \xc \setminus \{\gi, \Gi\}  &\greaterval{i} \xa \\
        \xc \setminus \{\gj, \Gj\}  &\greaterval{j} \xa \\
        \xa &\greatereqval{i} \xb \setminus \lstitem{i}{1}{\xb} \\
        \xb \cup \gj &\efxenvygg{j} \xc \setminus \gj         \end{align*}
    \end{definition}

    \begin{definition}[Stage B2ii for three agents]
        We say partition $\X=(\xa,\xb,\xc)$ is in \stageTwoPerfect\ if it is in \stageTwoNormal\ and:
        \begin{align*}
            \xb \setminus \lstitem{i}{1}{\xb} \greaterval{i} \xa
        \end{align*}
    \end{definition}

    Our proposed algorithm is similar to the one with four agent but far simpler. Once again, it starts by using the PR algorithm, but in order to guarantee a polynomial running time, it slightly modifies this algorithm so that the item removed from $Y_i$ and added to $Y_j$ is $g\in \arg\max_{x\in Y_i} v(Y_i \setminus x)$. 
    Using this modified version of the PR algorithm, we compute a partition 
    $(\xa,\xb)$ such that $\xa$ and $\xb$ are \efxf\ for some distinct agents $i$ and $j$, respectively. Using the output of the modified PR algorithm, our algorithm constructs the initial partition $\xp = (\xap =\emptyset,\xbp = \xa, \xcp = \xb)$ which is in \stageTwo\
    for three agents.
    Then, it follows the same steps as in \stageTwo\ (along with all its sub-stages for four agents), except there is no potential function and no agent $1$, so the operation relating to that agent will not be performed at all. Moreover, the bundle $\xb$ for the case of four agents is ignored here; note that in any transformation in the case with four agents, the second bundle always remains unchanged.

    At the start of each step, we will have some partition $\X = (\xa,\xb,\xc)$
    such that $\xb$ is \efxf\ for agent $i$ and $\xc$ is \efxf\ for some other agent $j$.
    Then, like the procedure in \stageTwo\ we 
    either construct an \eftxf\ partition or
    construct a partition
    $\xt = (\xa \cup x, \xbt, \xct)$ such that $\xbt$ is \efxf\ for some agent $i$
    and $\xct$ is \efxf\ for some other agent $j$. 
    At each step, the process consists of first swap-optimizing $(\xb,\xc)$
    w.r.t.\ agents $(i,j)$.
    Then, as in the \cref{PerfectExists}, we either find an \eftxf\ partition, 
    or we can find a partition $\xt$ in \stageTwoC\ or \stageTwoPerfect\ such that $\xat=\xa$.
    Finally, we execute the same as  \cref{Perfect} and \cref{solve perfect}.
    We remark that since we do not consider any agent as agent $1$, our procedure will not pass through \stageOne\ at all. 

    Since there is no agent $1$,
    the first bundle never loses an item, so this process can be repeated at most $m$ times,
    since there are $m$ goods. Also, every one of these steps runs in polynomial time.
    The only non obvious part of this claim is that the PR algorithm runs in polynomial time for two agents,
    which we prove next. Note that since there is no agent 1, no PR algorithm with more than two bundles will be executed. 

\subsection{Running Time of the Modified PR Algorithm}

    In the proof of the following lemma, for a valuation function $v$ and two subsets of items $S,T$
    we denote by $S \prec T$ whenever $v(S)<v(T)$ and we denote by $S\preceq T$ whenever $v(S)\leq v(T)$.

\begin{lemma}
    The modified PR algorithm runs in polynomial time 
    when the input is a set of two disjoint bundles 
    and a cancelable valuation function.
\end{lemma}
\begin{proof}
    Note that since valuations are cancelable, by \cref{k-min with k-least},
    when some bundle $S$ loses some good $x$  during the modified PR algorithm, $x$ is the least valued good in $S$.
    Suppose that during the PR algorithm, at some step, we have a set of bundles
    $(\ya,\yb)$ 
    and $x= \ell(\ya)$, such that $\ya\setminus x \succ \yb$ and $Y_1\setminus x \prec \yb \cup x$. 
    Then the algorithm removes $x$ from $\ya$ and adds it to $\yb$.

\begin{claim}
    Considering the above bundles $(\ya,\yb)$, after the transfer of $x$, 
    there will not be any step in the algorithm
    in which some good  $y$ is being transferred  such that $x \preceq y$.
\end{claim}
    
\begin{proof}
    Suppose on the contrary that this is not the case, and let $y$ be the first good being transferred after $x$ during the 
    execution of the algorithm, such that $x \preceq y$. 
    Then, for any good $z$ being transferred after $x$ and before $y$, if any, we have $z \prec x$. 
    Suppose that the bundles exactly before the transfer of $y$ are $(\yap,\ybp)$.
    Since we had $x = \ell(\ya)$, for every $g \in \ya \setminus x$, 
    we have $g \succeq x$, so these goods have not been transferred after $x$ and before $y$,
    hence, $\xa \setminus x \subseteq \yap$.
    Also since $\yap \sqcup \ybp = M$, we get that  
    $\ybp \subseteq \yb \cup x$. 
    In other words, all goods that are transferred between $x$ and $y$ should be transferred from $\yb\cup x$ to $\ya \setminus x$.
    
    We consider two cases based on which set, $\yap$ or $\ybp$, $y$ belongs to. 
     
    \vspace{5pt}
    \noindent $\mathbf{\bullet}$ \textbf{\caseA.} $\mathbf{y\in \yap}$:
        Since $y$ is being transferred from $\yap$, by the definition of PR algorithm, we have 
        that $y= \ell(\yap)$, so for every $g\in\yap$ we have that $g \succeq y \succeq x$,
        so $g \succeq x$. Since for every $z$ being transferred after $x$ and before $y$
        we have that $z \prec x$, we get that there is no good in $\yap$ that has been transferred from
        the other bundle to $\yap$ after
        the transfer of $x$, so $\yap = \ya \setminus x$, and $\ybp = \yb \cup x$.
        So we get $\yap = \ya \setminus x \prec \yb \cup x = \ybp$,
        which is a contradiction to the PR algorithm, since the algorithm transfers a good from the most valuable set to the least valuable one. 

    \vspace{5pt}
    \noindent $\mathbf{\bullet}$ \textbf{\caseB.} $\mathbf{y\in \ybp}$:
    Since for every good $z$ that has been transferred after $x$, we have that $z \prec x$,
    we get that $x \in \ybp$. Recall that $x \preceq y$. 
    if $x \prec y$, then $\ell(\yap) \neq y$. Therefore, $x$ has equal value to $y$, 
    which means that $\ybp \setminus y$ and $\ybp \setminus x$ 
    have the same value\footnote{This holds because if it was w.l.o.g. 
    $\ybp \setminus y \prec \ybp \setminus x$, then by subtracting $\ybp \setminus \{y,x\}$ from both sides, 
    cancelability would give $x\prec y$.}. Moreover, since
    $\ybp \subseteq \yb\cup x$, we get $\ybp\setminus x\subseteq \yb$, 
    which is translated to $\ybp \setminus x \preceq \yb$, 
    and therefore, $\ybp \setminus y \preceq \yb$. 
    Therefore, it holds:
    \begin{equation}
        \ybp \setminus y \preceq \yb \prec \ya \setminus x \preceq \yap\,, \notag
     \end{equation}
    where the last inequality is due to $\ya \setminus x \subseteq \yap$.
    So we have $\ybp \setminus y \prec \yap$, which is a contradiction, 
    since in the PR algorithm, when the algorithm removes some good $y$ from
    $\ybp$ and adds it to $\yap$, it should be that $\ybp \setminus y \succ \yap$.
    \end{proof}

    We now proceed with the lemma's proof.
    During execution of the modified PR algorithm, there can be
    at most $m$ removing good from $\ya$ and adding to $\yb$ in a row, since there are $m$ goods.
    Also, for the last good $x$, which is been removed from $\ya$ in a row, we  have  that
    $\ya \setminus x \succ \yb$, and if algorithm does not terminate it also holds that
    $\ya \setminus x \prec \yb \cup x$. Hence for every at most $m$ removal of goods from $\ya$ in a row,
    there exists at least one good (the last one) that according to the previous claim, will not come back to $\ya$
    ever again. Hence there could be at most $m^2$ removals of goods from $\ya$.
    Similarly, this holds for $\yb$ too, therefore the PR algorithm runs in polynomial time.
\end{proof}

\bibliographystyle{plainnat}
\bibliography{mybibliography}

\appendix

\section{Omitted Proofs}

\begin{definition}
    \cite{BCFF21}
    A valuation $v'$ is said to respect another valuation $v$ 
    if for every two bundles $S,T\subseteq M$ such that $v(S)>v(T)$
    it also holds that $v'(S)>v'(T)$.
\end{definition}

\begin{lemma}
\label{non-degenerate}
    \cite{ACGMM22} If $v$ is a monotone valuation, then there exists a non-degenerate valuation $v'$ that respects $v$.
\end{lemma}
\begin{proof}
    Let $M = \{g_1, g_2, \ldots, g_m\}$. We perturb valuation $v$ to $v'$.
    Let $\delta = \min_{S,T: \hspace{1mm} v(S) \ne v(T)}  |v(S)- v(T)|$, and let $\epsilon > 0$, 
    be such that $\epsilon 2^{m+1} < \delta$. Then let:

\begin{center}
$\forall S\subset M:  \hspace{3mm}  v_i'(S) = v_i(S) + \epsilon \sum_{g_j \in S}  2^j$.
\end{center}

Now suppose that  for two arbitrary sets $S,T$, we have $v(S)>v(T)$, then:
\begin{align*}
    v_i'(S) - v_i'(T) &= v_i(S) - v_i(T) + \epsilon (\sum_{g_j \in S\setminus T} 2^j  - \sum_{g_j \in T \setminus S} 2^j )  \\
    &\geq  \delta - \epsilon \sum_{g_j \in T \setminus S} 2^j  \\
    &\geq \delta - \epsilon(2^{m+1} -1) > 0
\end{align*}

    Hence, $v'$ respects $v$. Consider any two sets $S,T \subseteq M$ such that  $S \ne T$.
    If $v(S) \ne v(T)$,  we have 
    $v'(S) \ne v'(T)$, since $v'$ respects $v$. 
    If $v(S) = v(T)$, we have
    $v'(S) - v'(T) = \epsilon (\sum_{g_j \in S\setminus T} 2^j  - \sum_{g_j \in T \setminus S} 2^j ) \ne 0 $. Therefore, $v'$ is non-degenerate.
\end{proof}

\begin{lemma}
\label{generalize to degenerate}
    If $X = (X_1, X_2, \ldots, X_n)$ is an EF$k$X allocation for agents $1,2,\ldots, n$ with valuations
    $V' =(v_1', v_2', \ldots, v_n')$, respectively, and if $V=(v_1, v_2, \ldots, v_n)$ is a set of valuations that for every
    $i \in [n]$, valuation $v_i'$ respects valuation $v_i$,
    then $X$ is also an EF$k$X allocation for agents $1,2,\ldots,n$ with valuations $v_1,v_2,\ldots,v_n$, respectively.
\end{lemma}

\begin{proof}
    Let us assume that $\X$ is an EF$k$X allocation with valuations $V'$, 
    and not an EF$k$X allocation with valuations $V$. 
    Then there exist $i,j$, and $g_1,g_2,...,g_p \in X_j$, 
    with $p = \min(k, |X_j|)$, 
    such that $v_i(X_j \setminus \{g_1,\ldots, g_p\}) > v_i(X_i)$. 
    In that case, Since $v'_i$ respects $v_i$, we have
    $v'_i(X_j \setminus \{g_1,\ldots, g_p\})> v'_i(X_i)$,
    implying that $\X$ is not an EF$k$X allocation
    with valuations $V'$ as well, which is a contradiction.
\end{proof}

\begin{corollary}
\label{non degenerate}
    When agent $1$ has a monotone valuation $v$, and agents $\{2,3,4\}$ have valuations $v_1,v_2,v_3$,
    by \cref{non-degenerate}, there exists a set of non-degenerate valuations $v'_1,v_2,v_3,v_4$, 
    that respect valuations $v_1,v_2,v_3,v_4$, respectively. 
    So, we can find an \eftx allocation with these new valuations, and then by \cref{generalize to degenerate},
    this allocation is an \eftx allocation with original valuations too.
\end{corollary}

% \section*{Acknowledgements}
% Vasilis Gkatzelis was partially supported by  NSF CAREER award CCF-2047907. The research project is implemented in the framework of H.F.R.I call “Basic research Financing (Horizontal support of all Sciences)” under the National Recovery and Resilience Plan “Greece 2.0” funded by the European Union-–NextGenerationEU (H.F.R.I. Project Number:15635).

\end{document}

%% file: Macros.tex
\newcommand{\itm}{good}
\newcommand{\itms}{goods}

\newcommand{\bestkvalue}{best-k-value}
\newcommand{\bestkremainer}{best-k-remaining \itms}

\newcommand{\with}{relative to}

\newcommand{\unassignedItms}{S}

\newcommand{\argmax}{H}
\newcommand{\argmin}{L}

\newcommand{\eftx}{EF$2$X}

\newcommand{\efxf}{EFX-feasible}
\newcommand{\bstf}{EFX-best}
\newcommand{\bsttwof}{EF$2$X-best}
\newcommand{\eftxf}{EF$2$X-feasible}
\newcommand{\bstkf}{EFkX-best}
\newcommand{\efkxf}{EFkX-feasible}

\newcommand{\efxenvy}{EFX-envy}
\newcommand{\efxenvies}{EFX-envies}
\newcommand{\eftxenvy}{EF$2$X-envy}
\newcommand{\eftxenvies}{EF$2$X-envies}

\newcommand{\efxfset}{EFXF}
\newcommand{\bstfset}{EFXBest}
\newcommand{\eftxfset}{EF$2$XF}
\newcommand{\bsttfset}{EF$2$XBest}

\newcommand{\isefxf}[3]{$#1 \in \text{\efxfset}_{#2}(#3)$}
\newcommand{\isnotefxf}[3]{$#1 \not\in \text{\efxfset}_{#2}(#3)$}
\newcommand{\isbstf}[3]{$#1 \in \text{\bstfset}_{#2}(#3)$}

\newcommand{\isefxffor}[3]{$#1 \in \text{\efxfset}_{#2}(#3)$}
\newcommand{\isnotefxffor}[3]{$#1 \not\in \text{\efxfset}_{#2}(#3)$}
\newcommand{\isbstffor}[3]{$#1 \in \text{\bstfset}_{#2}(#3)$}
\newcommand{\isnotbstffor}[3]{$#1 \not \in \text{\bstfset}_{#2}(#3)$}
\newcommand{\iseftxffor}[3]{$#1 \in \text{\eftxfset}_{#2}(#3)$}
\newcommand{\arenoteftxffor}[3]{$#1 \not\in \text{\eftxfset}_{#2}(#3)$}
\newcommand{\isbsttwoffor}[3]{$#1 \in \text{\bsttfset}_{#2}(#3)$}
\newcommand{\arenotbsttwoffor}[3]{$#1 \not\in \text{\bsttfset}_{#2}(#3)$}
\newcommand{\areefxffor}[3]{$#1 \in \text{\efxfset}_{#2}(#3)$}
\newcommand{\arenotefxffor}[3]{$#1 \not\in \text{\efxfset}_{#2}(#3)$}

\newcommand{\areeftxffor}[3]{$#1 \in \text{\eftxfset}_{#2}(#3)$}

\newcommand{\good}{in-stage}
\newcommand{\goodpart}{in-stage partition}

\newcommand{\stageOne}{stage A}
\newcommand{\stageOnei}{stage Ai}
\newcommand{\stageOneii}{stage Aii}
\newcommand{\StageOnei}{Stage Ai}
\newcommand{\StageOneii}{Stage Aii}
\newcommand{\stageTwo}{stage B}
\newcommand{\stageTwoA}{stage B1}
\newcommand{\stageTwoNormal}{stage B2}
\newcommand{\stageTwoC}{stage B2i}
\newcommand{\stageTwoPerfect}{stage B2ii}

\newcommand{\normalize}{swap-optimize}
\newcommand{\normalized}{swap-optimal}
\newcommand{\normalizationprocess}{swap-optimization}
\newcommand{\normalizationality}{swap-optimality}

\newcommand{\StageOne}{Stage A}
\newcommand{\StageTwo}{Stage B}
\newcommand{\StageTwoA}{Stage B1}
\newcommand{\StageTwoNormal}{Stage B2}
\newcommand{\StageTwoC}{Stage B2i}
\newcommand{\StageTwoPerfect}{Stage B2ii}

\newcommand{\NormalizationProcess}{Swap-Optimization}

\newcommand{\X}{\mathbf{X}}
\newcommand{\Y}{\mathbf{Y}}
\newcommand{\Z}{\mathbf{Z}}

\newcommand{\xp}{\mathbf{X'}}
\newcommand{\xz}{\mathbf{X''}}
\newcommand{\xt}{\tilde{\X}}

\newcommand{\yp}{\Y'}

\newcommand{\ya}{Y_1}
\newcommand{\yb}{Y_2}
\newcommand{\yc}{Y_3}
\newcommand{\yd}{Y_4}
\newcommand{\yap}{Y'_1}
\newcommand{\ybp}{Y'_2}
\newcommand{\ycp}{Y'_3}
\newcommand{\ydp}{Y'_4}

\newcommand{\partt}{(X_1,X_2,X_3)}
\newcommand{\parttp}{(X_1^\prime,X_2^\prime,X_3^\prime)  }

\newcommand{\minone}{v_1(X_1)}
\newcommand{\mintwo}{\min_{r\in [2]} v_1(X_r)}
\newcommand{\minthree}{\min_{r\in [3]} v_1(X_r)}
\newcommand{\minfour}{\min_{r \in [4]} v_1(X_r)}
\newcommand{\minonep}{v_1(X'_1)}
\newcommand{\mintwop}{\min_{r\in [2]} v_1(X'_r)}
\newcommand{\minthreep}{\min_{r\in [3]} v_1(X'_r)}
\newcommand{\minfourp}{\min_{r \in [4]} v_1(X'_r)}

\newcommand{\minfourz}{\min_{r \in [4]} v_1(X''_r)}
\newcommand{\minonet}{v_1(\tilde{X_1})}

\newcommand{\mintwoy}{\min_{r\in [2]} v_1(Y_r)}

\newcommand{\xa}{X_1}
\newcommand{\xb}{X_2}
\newcommand{\xc}{X_3}
\newcommand{\xd}{X_4}

\newcommand{\xat}{\tilde{X_1}}
\newcommand{\xbt}{\tilde{X_2}}
\newcommand{\xct}{\tilde{X_3}}
\newcommand{\xdt}{\tilde{X_4}}

\newcommand{\xap}{X^\prime_1}
\newcommand{\xbp}{X^\prime_2}
\newcommand{\xcp}{X^\prime_3}
\newcommand{\xdp}{X^\prime_4}

\newcommand{\xaz}{X^{\prime \prime}_1}
\newcommand{\xbz}{X^{\prime \prime}_2}
\newcommand{\xcz}{X^{\prime \prime}_3}
\newcommand{\xdz}{X^{\prime \prime}_4}

\newcommand{\xab}{$(X_1,X_2)$ }

\newcommand{\lst}[3]{\ell_{#1}^{#2}(#3)}
\newcommand{\hst}[3]{h_{#1}^{#2}(#3)}

\newcommand{\lstitem}[3]{\ell_{#1}^{#2}(#3)}
\newcommand{\hstitem}[3]{h_{#1}^{#2}(#3)}
\newcommand{\biglstitem}[3]{\ell_{#1}^{#2}\Big(#3\Big)}
\newcommand{\bigerlstitem}[3]{\ell_{#1}^{#2}\big(#3\big)}
\newcommand{\bigerhstitem}[3]{h_{#1}^{#2}\big(#3\big)}

\newcommand{\hi}{\lst{i}{1}{X_4}}
\newcommand{\Hi}{\lst{i}{2}{X_4}}
\newcommand{\hj}{\lst{j}{1}{X_4}}
\newcommand{\Hj}{\lst{j}{2}{X_4}}
\newcommand{\hu}{\lst{u}{1}{X_4}}

\newcommand{\hip}{\lst{i}{1}{X'_4}}
\newcommand{\Hip}{\lst{i}{2}{X'_4}}
\newcommand{\hjp}{\lst{j}{1}{X'_4}}
\newcommand{\Hjp}{\lst{j}{2}{X'_4}}
\newcommand{\hup}{\lst{u}{1}{X'_4}}

\newcommand{\hiz}{\lst{i}{1}{X''_4}}
\newcommand{\Hiz}{\lst{i}{2}{X''_4}}
\newcommand{\hjz}{\lst{j}{1}{X''_4}}
\newcommand{\Hjz}{\lst{j}{2}{X''_4}}

\newcommand{\hiy}{\lst{i}{1}{Y_4}}

\newcommand{\gi}{\lst{i}{1}{X_3}}
\newcommand{\Gi}{\lst{i}{2}{X_3}}
\newcommand{\gj}{\lst{j}{1}{X_3}}
\newcommand{\Gj}{\lst{j}{2}{X_3}}

\newcommand{\gip}{\lst{i}{1}{X'_3}}
\newcommand{\Gip}{\lst{i}{2}{X'_3}}
\newcommand{\gjp}{\lst{j}{1}{X'_3}}
\newcommand{\Gjp}{\lst{j}{2}{X'_3}}
\newcommand{\gup}{\lst{u}{1}{X'_3}}

\newcommand{\giz}{\lst{i}{1}{X''_3}}
\newcommand{\Giz}{\lst{i}{2}{X''_3}}
\newcommand{\gjz}{\lst{j}{1}{X''_3}}
\newcommand{\Gjz}{\lst{j}{2}{X''_3}}

\newcommand{\lowerval}[1]{\prec_{#1}}
\newcommand{\lowereqval}[1]{\preceq_{#1}}
\newcommand{\greaterval}[1]{\succ_{#1}}
\newcommand{\greatereqval}[1]{\succeq_{#1}}
\newcommand{\efxenvyll}[1]{\llcurly_{#1}}
\newcommand{\efxenvygg}[1]{\ggcurly_{#1}}
\newcommand{\doesnotefxenvyll}[1]{\centernot\llcurly_{\!\!#1~}}
\newcommand{\doesnotefxenvygg}[1]{\centernot\ggcurly_{\!\!#1~}}

\newcommand{\caseA}{Case 1}
\newcommand{\caseB}{Case 2}

\newcommand{\subcaseA}{Sub-case 1}
\newcommand{\subcaseB}{Sub-case 2}
\newcommand{\subcaseC}{Sub-case 3}

\newcommand{\donef}{
Then we can construct a new partition $\xt$ that is either \eftxf \ or \good\
with $\phi(\xt)>\phi(X)$.}

\newcommand{\donew}{
Therefore by \cref{ToGoodPart3}, we can construct a new partition $\xt$ that is either \eftxf \ or \good \ 
with $\phi(\xt)\geq \mintwop > \phi(\X)$.}